\documentclass{article}

\usepackage{fullpage}
\usepackage{authblk}
\usepackage{enumerate} 
\usepackage[hidelinks]{hyperref}

\usepackage[american]{babel}	
\usepackage[utf8]{inputenc}		
\usepackage[T1]{fontenc}	
\usepackage[english=american]{csquotes}	
\usepackage[numbers]{natbib}
\usepackage{graphicx}		
\usepackage{amsmath}		
\usepackage{amssymb}
\usepackage{amsthm}
\usepackage{mathtools}
\usepackage{amsfonts}    	
\usepackage{xcolor}
\usepackage[format=plain,labelfont=bf]{caption}
\usepackage{subcaption}
\usepackage{rotating}
\usepackage{float}
\usepackage{booktabs}
\usepackage{array}

\newcommand\blfootnote[1]{%
  \begingroup
  \renewcommand\thefootnote{}\footnote{#1}%
  \addtocounter{footnote}{-1}%
  \endgroup
}

\newcommand{\cS}{{\mathcal S}}
\newcommand{\cH}{{\mathcal H}}
\newcommand{\bP}{{\mathbb P}}
\newcommand{\bE}{{\mathbb E}}

\theoremstyle{definition}

\theoremstyle{plain}
\newtheorem{lemma}{Lemma}
\newtheorem{theorem}{Theorem}
\newtheorem{proposition}{Proposition}
\newtheorem{corollary}{Corollary}

\title{Effects of discordance between species and gene trees on
phylogenetic diversity conservation}

\author[1,2$\ast$]{Kristina Wicke}
\author[3]{Mareike Fischer}
\author[4]{Laura Kubatko}

\affil[1]{Department of Mathematics, The Ohio State University, Columbus OH, USA}
\affil[2]{Department of Mathematical Sciences, New Jersey Institute of Technology, Newark NJ, USA}
\affil[3]{Institute of Mathematics and Computer Science, University of Greifswald, Greifswald, Germany}
\affil[4]{Department of Statistics, Department of Evolution, Ecology, and Organismal Biology, The Ohio State University, Columbus OH, USA}

\date{}

\begin{document}
\maketitle

\begin{abstract} 
\noindent Phylogenetic diversity indices such as the Fair Proportion (FP) index are frequently discussed as prioritization criteria in biodiversity conservation. They rank species according to their contribution to overall diversity by taking into account the unique and shared evolutionary history of each species as indicated by its placement in an underlying phylogenetic tree. 
Traditionally, phylogenetic trees were inferred from single genes and the resulting gene trees were assumed to be a valid estimate for the species tree, i.e., the ``true'' evolutionary history of the species under consideration. However, nowadays it is common to sequence whole genomes of hundreds or thousands of genes, and it is often the case that conflicting genealogical histories exist in different genes throughout the genome, resulting in discordance between individual gene trees and the species tree. Here, we analyze the effects of gene and species tree discordance on prioritization decisions based on the FP index. 
In particular, we consider the ranking order of taxa induced by (i) the FP index on a species tree, and (ii) the expected FP index across all gene tree histories associated with the species tree. On one hand, we show that for particular tree shapes, the two rankings always coincide. On the other hand, we show that for all leaf numbers greater than or equal to five, there exist species trees for which the two rankings differ. 
Finally, we illustrate the variability in the rankings obtained from the FP index across different gene tree and species tree estimates for an empirical multilocus mammal data set.
\end{abstract}


\blfootnote{$^\ast$Corresponding author\\ \textit{Email address:} \url{kristina.wicke@njit.edu}}

\section{Introduction} \label{Sec_Intro}
Phylogenetic diversity indices such as the Fair Proportion (FP) index \citep{Isaac2007,Redding2003}, Equal Splits (ES) index \citep{Redding2006}, and Shapley value \citep{Haake2008}, have become popular tools for prioritizing species for conservation (e.g., \citet{Isaac2007, Redding2006, Redding2008,Redding2014, Vanewright1991, Vellend2011}). They quantify the importance of species for overall biodiversity based on their placement in an underlying phylogenetic tree and can thus, next to other criteria such as threat status, serve as prioritization criteria in conservation decisions. While these indices differ in their definitions and properties, they all depend on the underlying phylogenetic tree and its branch lengths. In particular, different phylogenetic trees for the same set of species or identical phylogenetic trees with different branch lengths, may result in different prioritization orders obtained by any of these phylogenetic diversity indices. 

Until recently, phylogenetic trees were usually inferred from single genes and the resulting gene trees were assumed to be a valid estimate for the species tree, i.e., the ``true'' evolutionary history of the species trees under consideration. However, with the advent of rapid sequencing technologies in recent years, it is now common to sequence whole genomes composed of hundreds or thousands of genes, and it is well known that due to various biological phenomena, such as lateral gene transfer and incomplete lineage sorting, conflicting genealogical histories exist in different genes throughout the genome, resulting in discordance between individual gene trees and the species tree (e.g., \citet{Maddison1997, Nichols2001, Pamilo1988}). 

In this paper, we analyze the effects of gene and species tree discordance on prioritization decisions based on the FP index. We chose the FP index as it is employed (under the name `evolutionary distinctiveness') in the so-called ``EDGE of Existence'' project established by the Zoological Society of London, a conservation initiative focusing specifically at threatened species that represent a high amount of unique evolutionary history (\cite{Isaac2007}; see also \url{https://www.edgeofexistence.org/}). Moreover, while the FP index itself lacks a direct biological motivation, it coincides with the so-called Shapley value, a result that was first shown by \citet{Fuchs2015}. The Shapley value, originally a concept from cooperative game theory introduced by Lloyd Shapley \citep{Shapley1953}, on the other hand, reflects the average contribution of a species to overall diversity and is characterized by certain desirable axioms (see, e.g., \citet{Haake2008}) that motivate its use (and hence that of the FP index) as a prioritization criterion in conservation planning.

Moreover, we focus on gene tree-species tree incongruence caused by incomplete lineage sorting by employing the multispecies coalescent model \citep{Degnan2009}. Under this model, given a species tree (an ultrametric phylogenetic tree with branch lengths), a distribution on possible gene tree histories evolving within the species tree is obtained, and the species tree as well as the distribution of gene tree histories are the basis for our analysis. More precisely, we compare the prioritization order of taxa induced by the FP index on the species tree with the prioritization order obtained when averaging the FP index over all gene tree histories evolving within the species tree. 

The structure of this paper is as follows. We first introduce all relevant notation and concepts. We then analyze circumstances under which the FP index on a given species tree and the expected FP index across all gene tree histories evolving within the species tree induce the same ranking (Theorem \ref{Thm_Caterpillar_Pseudocaterpillar}). Afterwards, we show that for all $n \geq 5$, there exists a species tree on $n$ leaves such that the two rankings (i.e., the ranking obtained from the species tree and the ranking obtained from averaging over all gene tree histories) differ (Theorem \ref{Thm_rankswap}). Finally, we complement our theoretical results by analyzing the variability in the rankings obtained from the FP index across different gene trees and species trees for an empirical multilocus mammal data set. We end by discussing our results and indicating directions for future research.

\section{Notation and preliminaries} \label{Sec_Prelim}
To formally state our results, we need some terminology and notation.
Throughout this manuscript, $X$ denotes a non-empty finite set (of taxa or species). If not stated otherwise, $X=\{x_1, \ldots, x_n\}$.

\paragraph*{Species trees, gene tree histories, and probabilities under the multispecies coalescent model.}
\noindent A \emph{species tree} $\cS=(S,\boldsymbol{\tau})$ on $X=\{x_1, \ldots, x_n\}$ is a pair consisting of a rooted binary phylogenetic tree $S$ with leaf set $X$ (i.e., a rooted binary tree whose leaves are identified with $X$) and a vector of interval lengths $\boldsymbol{\tau}=(\tau_0, \tau_1, \ldots, \tau_{n-1})$, where $\tau_i$ denotes the length of the interval between the $i^{\text{th}}$ and $(i+1)^{\text{th}}$ speciation event in $S$ measured in \emph{coalescent units}. Note that we assume here that no two speciation events in $S$ occur at exactly the same time (this is motivated by the fact that under typical models for species trees, the probability for two speciation events to occur simultaneously is zero). Moreover, note that the length of $\tau_{n-1}$, i.e., the length of the interval above the root of $S$, is infinite. Finally, note that the definition of $\boldsymbol{\tau} $ as a vector of interval lengths between speciation times implies that $\cS=(S, \boldsymbol{\tau})$ is \emph{equidistant} or \emph{ultrametric}, meaning that the length of the path from the root of $\cS$ to any leaf of the species tree is the same. This is also commonly referred to as the \emph{molecular clock assumption}. Moreover, note that we sometimes refer to the length of the path from the root of $\cS$ to any leaf as the \emph{height} of $\cS$. 

Given a species tree $\cS=(S,\boldsymbol{\tau})$, we use the concept of a \emph{gene tree history $h$} associated with $\cS$ to indicate a sequence of coalescent events as well as the populations and time intervals in which they occur. Informally, a gene tree history $h$ can be thought of as a phylogenetic tree with leaf set $X$ evolving within $\cS$ (Figure \ref{Fig_5Leaves_Notation}). Note that the set of gene tree histories associated with $\cS$ is finite, and we denote it by $\mathcal{H}(\cS)$.

Moreover, we associate with each history $h$ a vector of coalescent times $\boldsymbol{t}_h=(t_1, \ldots, t_{n-1})$, where, moving back in time, two distinct lineages coalesce or merge into a single lineage. For example, referring to Figure \ref{Fig_5Leaves_Notation}, at time $t_1$, the lineages corresponding to $x_1$ and $x_2$ merge into a single lineage. Note that there are infinitely many choices for the values $t_1, \ldots, t_{n-1}$ satisfying the constraints of a history $h$ evolving within $\cS$, and in the following we will often integrate over all possible coalescent times consistent with $h$. Moreover, note that we assume that no two coalescences occur at exactly the same time, i.e., $t_i \neq t_j$ for all $i,j \in \{1, \ldots, n-1\}$. 
Thus the coalescent times induce a ranking of the interior vertices of a gene tree history $h$ according to their distance from the root, where the interior vertex with the largest distance (i.e., the vertex corresponding to the most recent coalescence event) is assigned rank 1 and the root itself is assigned rank $n-1$. In particular, we consider two histories to be distinct if their coalescent times induce different rankings of the interior vertices (Figure \ref{Fig_h1h2}). Formally, this corresponds to the notion of \emph{ranked gene tree histories} \citep{Degnan2012, Degnan2012a} (similarly, \emph{ranked species trees} are defined as species tree topologies together with an order on the speciation events). However, for brevity, we speak of gene tree histories or histories for short.  In addition, when we refer to a history $h$, we usually mean $h$ and its associated vector of coalescent times $\boldsymbol{t}_h$ as well as the information in which time interval each coalescent event occurred. Note that we explicitly keep track of the time intervals and, next to the induced rankings of interior vertices, distinguish histories based on this information. For instance, considering the top part of Figure \ref{Fig_h1h2}, the scenario where the coalescent event $t_1$ occurs in time interval $\tau_1$ (rather than $\tau_2$) in the population ancestral to $x_1$ and $x_2$, is considered a different history, $h_3$ say, that is distinct from $h_1$.

\begin{figure}[htbp]
    \centering
    \includegraphics[scale=0.175]{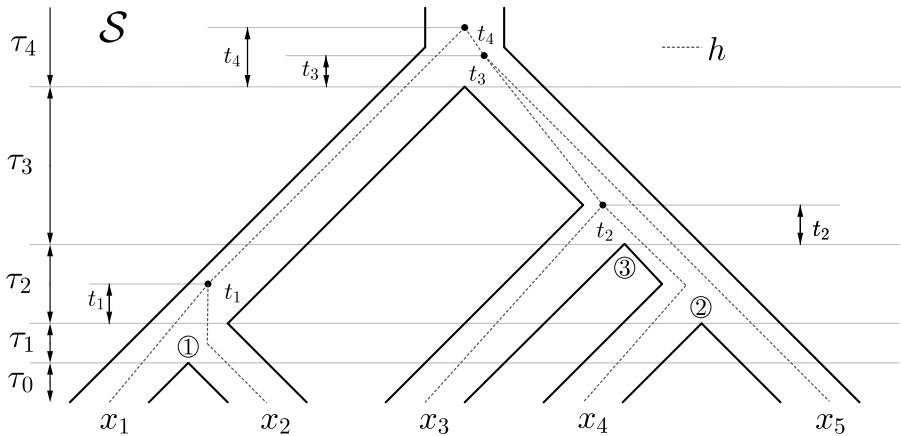}
    \caption{Species tree $\cS=(S,\boldsymbol{\tau})$ on $X=\{x_1, \ldots, x_5\}$ and history $h$ evolving within $\cS$. Note that $\cS$ and $h$ depict different evolutionary relationships. Intervals of time between speciation events (horizontal lines) are denoted by 
    the $\tau_i$ (left side of the figure), while times corresponding to coalescent events (black dots) are denoted by the $t_i$ (coalescent times are measured from the most recent speciation event). Circled numbers indicate a ranking on the speciation events with  \raisebox{.5pt}{\textcircled{\raisebox{-.9pt} {1}}} being the most recent and \raisebox{.5pt}{\textcircled{\raisebox{-.9pt} {3}}} the most ancient speciation event, and $\cS$ together with this ordering defines a ranked species tree.}
    \label{Fig_5Leaves_Notation}
\end{figure}

\begin{figure}[htbp]
    \centering
    \includegraphics[scale=0.12]{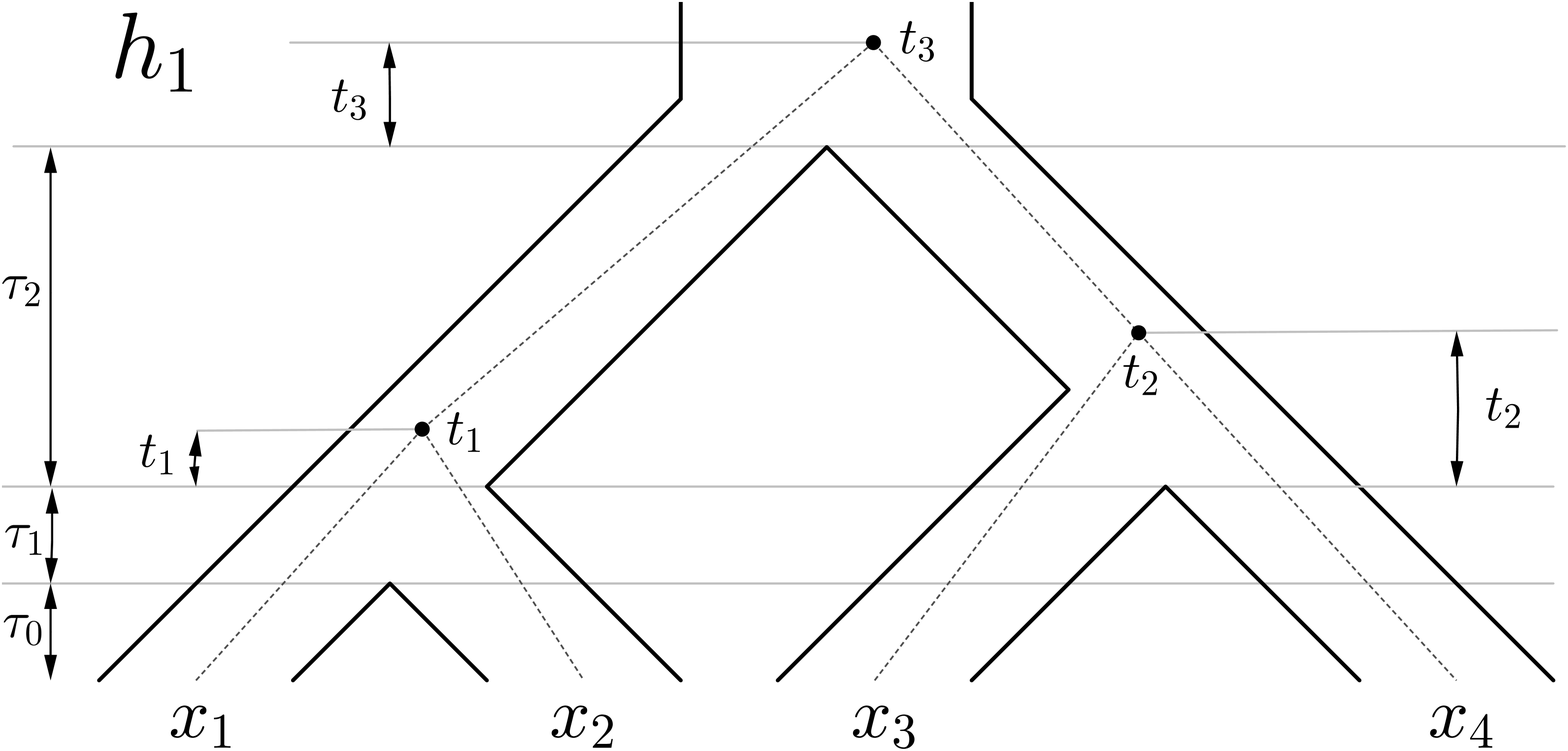} 
    \includegraphics[scale=0.12]{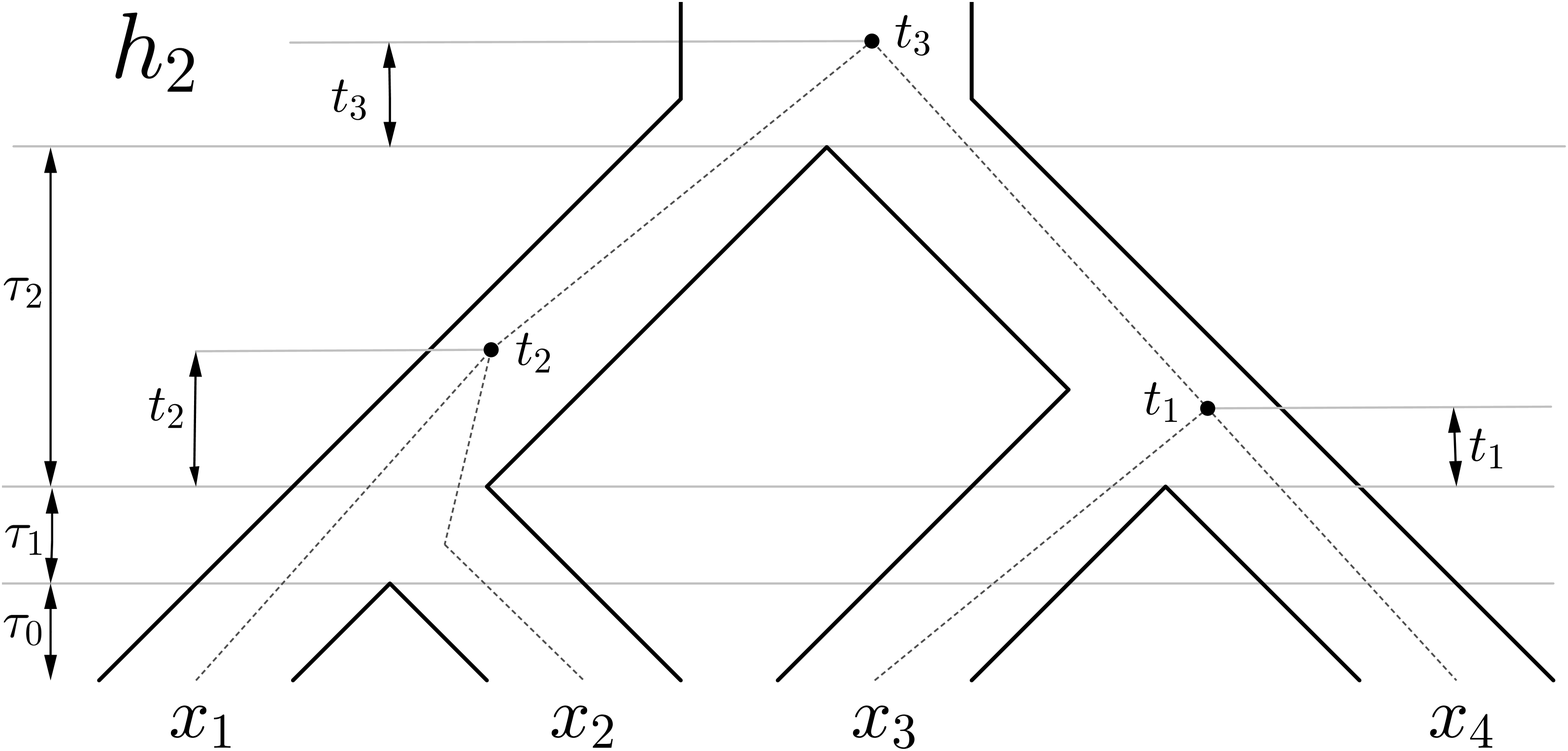}
    \caption{Two distinct histories $h_1$ and $h_2$ evolving within the same species tree. Note that $h_1$ and $h_2$ have the same topology and their coalescent events occur in the same time intervals of the species tree, but the rankings of the interior vertices of $h_1$, respectively $h_2$, induced by vectors of coalescent times $\boldsymbol{t}_{h_1}$ and $\boldsymbol{t}_{h_2}$ differ. Moreover, both $h_1$ and $h_2$ are such that there are two coalescence events in interval $\tau_2$. By subdividing this interval into three subintervals, namely $[0,t_1], [t_1,t_2]$, and $[t_2, \tau_2]$, and multiplying densities contributed by subintervals, we e.g. obtain $f_{h_1}(h_1,\boldsymbol{t}_{h_1}) = e^{- \tau_1} \cdot e^{-t_1} \cdot e^{-t_1} \cdot e^{-(t_2-t_1)} \cdot e^{-t_3} = e^{-\tau_1-t_1-t_2-t_3}$.}
    \label{Fig_h1h2}
\end{figure}

For a given species tree $\cS=(S,\boldsymbol{\tau})$ and a particular (ranked gene tree) history $h \in \cH(\cS)$ with associated vector of coalescent times $\boldsymbol{t}_h$, the probability of observing $h$ given $\cS$ under the multispecies coalescent model can be found by integrating over the \emph{joint density of coalescence times}, which we denote by $f_h(h,\boldsymbol{t}_h)$. We refer the reader to \citet{Degnan2012,Degnan2012a} for a careful derivation of this density. Essentially, it is a product of exponential terms arising from the occurrence of exponential distributions as waiting times to coalescent events. More precisely, for a sample of $k$ lineages from a single population, the waiting time to coalescence is exponentially distributed with rate $\binom{k}{2}$. The probability that there is no coalescence for a sample of $k$ lineages in a time interval of length $\tau$ is then $e^{- \binom{k}{2} \tau}$. Based on this, the density $f_h(h,\boldsymbol{t}_h)$ is obtained by multiplying exponential probabilities across the branches and time intervals of $\cS$. If there are $m_i \geq 0$ coalescences in an interval $\tau_i$ of $\cS$, the interval is subdivided into $m_i+1$ subintervals and densities contributed by the subintervals are multiplied (again, see \citet{Degnan2012,Degnan2012a} for a more rigorous definition; see also  Figure~\ref{Fig_h1h2} for an example). 

\noindent As an example, consider the species tree $\cS=(S,\boldsymbol{\tau})$ and the gene tree history $h \in \cH(\cS)$ depicted in Figure \ref{Fig_5Leaves_Notation}. First consider the ``left'' subtree of $\cS$ containing leaves $x_1$ and $x_2$. Here, $x_1$ and $x_2$ fail to coalesce in an interval of length $\tau_1$, contributing a factor of $e^{-\tau_1}$ to the density $f_h(h,\boldsymbol{t}_h)$. These lineages then coalesce at time $t_1$ within the next interval, contributing a factor of $e^{-t_1}$ to the density. Similarly, for the ``right'' subtree of $\cS$, lineages $x_4$ and $x_5$ fail to coalesce in an interval of length $\tau_2$ yielding a factor of $e^{-\tau_2}$. The interval $\tau_3$ then contains a coalescent event at time $t_2$, which contributes the term $e^{-3t_2}$ (since three lineages were available to coalesce). The remaining two lineages fail to coalesce for the remainder of this interval, contributing a term of $e^{-(\tau_3 - t_2)}$. Finally, considering $\tau_4$, the most recent coalescent event occurs at time $t_3$, thus contributing a factor of $e^{-3t_3}$ to the density, while the final coalescent event occurs $t_4-t_3$ time units later, contributing a term of $e^{-(t_4-t_3)}$. Thus, in this case, we have
$$f_h(h,\boldsymbol{t}_h) = e^{-\tau_1}e^{-t_1}  e^{-\tau_2}  e^{-3 t_2} e^{-(\tau_3-t_2)} e^{-3t_3} e^{-(t_4-t_3)} = e^{-\tau_1 - \tau_2 - \tau_3  - t_1 - 2t_2 - 2t_3 - t_4}.$$
Finally, the limits of integration are determined by the vector of coalescent times $\boldsymbol{t}_h$ and the vector of interval lengths $\boldsymbol{\tau}$. In the example discussed above, $t_1$ can range from $0$ to $\tau_2$, $t_2$ can range from $0$ to $\tau_3$, $t_3$ can range from $0$ to $t_4$, and $t_4$ can range from 0 to infinity. This leads to
\begin{align*}
 \bP(h \vert \cS)
    &= \int f_h(h,\boldsymbol{t}_h) \, d\boldsymbol{t}_h \\
     &= \int\limits_{t_4=0}^{\infty} \, \int\limits_{t_3=0}^{t_4} \, \int\limits_{t_2=0}^{\tau_3} \, \int\limits_{t_1=0}^{\tau_2} \, e^{-\tau_1 - \tau_2 - \tau_3  - t_1 - 2t_2 - 2t_3 - t_4} \, dt_1 \, dt_2 \, dt_3 \, dt_4  \\
     &= \frac{1}{6} e^{- \tau_1 - 2 \tau_2 - 3 \tau_3} -  \frac{1}{6} e^{- \tau_1 - \tau_2 - 3 \tau_3} -  \frac{1}{6} e^{- \tau_1 - 2 \tau_2 - \tau_3} +  \frac{1}{6} e^{- \tau_1 -  \tau_2 - \tau_3}. 
\end{align*}

\paragraph*{Exchangeability.}
Considering the coalescent process for some species tree $\cS$ on $X$ and a time interval $\tau$, we call two leaves, $x_i$ and $x_j$ say, \emph{exchangeable} (with respect to $\tau$) if they ``enter'' interval $\tau$ as single lineages. In particular, from this point onward (back in time), $x_i$ and $x_j$ are ``indistinguishable'' and every history $h$ with the property that neither $x_i$ nor $x_j$ coalesce before $\tau$ either contains $[x_i,x_j]$ as a cherry, or if such a history $h'$ does not contain the cherry $[x_i,x_j]$, there exists a history $h''$ which only differs from $h'$ by a permutation of $x_i$ and $x_j$ and such that $\boldsymbol{t}_{h'}=\boldsymbol{t}_{h''}$ and $f_{h'}(h',\boldsymbol{t}_{h'}) = f_{h''}(h'',\boldsymbol{t}_{h''}) $.

As an example, for the species tree $\cS$ shown in Figure \ref{Fig_5Leaves_Notation}, leaves $x_3$, $x_4$, and $x_5$ are exchangeable with respect to interval $\tau_3$ if they enter this interval as single lineages, as is the case for the history depicted. The history $h$ depicted contains $[x_3,x_4]$ as a cherry, but there also exist histories, $h'$ and $h''$ say, where $x_3$ and $x_4$ do not form a cherry, but such that $h'$ and $h''$ only differ by a permutation of $x_3$ and $x_4$, and such that $\boldsymbol{t}_{h'}=\boldsymbol{t}_{h''}$ and $f_{h'}(h',\boldsymbol{t}_{h'}) = f_{h''}(h'',\boldsymbol{t}_{h''}) $ (for instance, $h'$ could be the history identical to $h$ except that the coalescence at time $t_2$ involves $x_3$ and $x_5$, in which case the ``matching'' history $h''$ would be the history identical to $h$ except that the coalescence at time $t_2$ involves $x_4$ and $x_5$).

This exchangeability property and the idea of pairing matching histories will be crucial for many proofs and we will sometimes also consider exchangeability of subtrees (rather than single leaves) defined in an analogous way.

\paragraph*{Lowest common ancestors and subtrees.} 
Given a species tree $\cS$ (a history $h$) with root $\rho$ and two leaves $x_i,x_j \in X$, we refer to the last vertex in $\cS$ ($h$) that lies on both the path from $\rho$ to $x_i$ as well as the path from $\rho$ to $x_j$ as the \emph{lowest common ancestor} of $x_i$ and $x_j$ and denote it as $lca(x_i,x_j)$. Moreover, given a vertex $v$ of $\cS$ ($h$), we refer to the subtree of $\cS$ ($h$) rooted at vertex $v$ as a \emph{pendant subtree} of $\cS$ ($h$).

\paragraph*{Caterpillar and pseudocaterpillar trees.} 
Two particular tree topologies that will be relevant in the following are the \emph{caterpillar tree} and the \emph{pseudocaterpillar tree} (Figure \ref{Fig_CatPseudoCat}). First, recall that a pair of leaves, say $x_i$ and $x_j$, of a phylogenetic $X$-tree is called a cherry, denoted $[x_i,x_j]$, if $x_i$ and $x_j$ share a common parent. Now, let $n \geq 2$. Then, a phylogenetic $X$-tree on $X=\{x_1, \ldots, x_n\}$ is called a \emph{caterpillar tree} if and only if it has precisely one cherry. Moreover, if $n \geq 4$, a phylogenetic $X$-tree on $X=\{x_1, \ldots, x_n\}$ is called a \emph{pseudocaterpillar tree} (adapted from \cite{Degnan2012}\footnote{Note that \cite{Degnan2012} defined pseudocaterpillar trees for $n \geq 5$.}) if it has precisely two cherries and those two cherries form a pendant subtree of size four. 

\begin{figure}[htbp]
    \centering
    \includegraphics[scale=0.195]{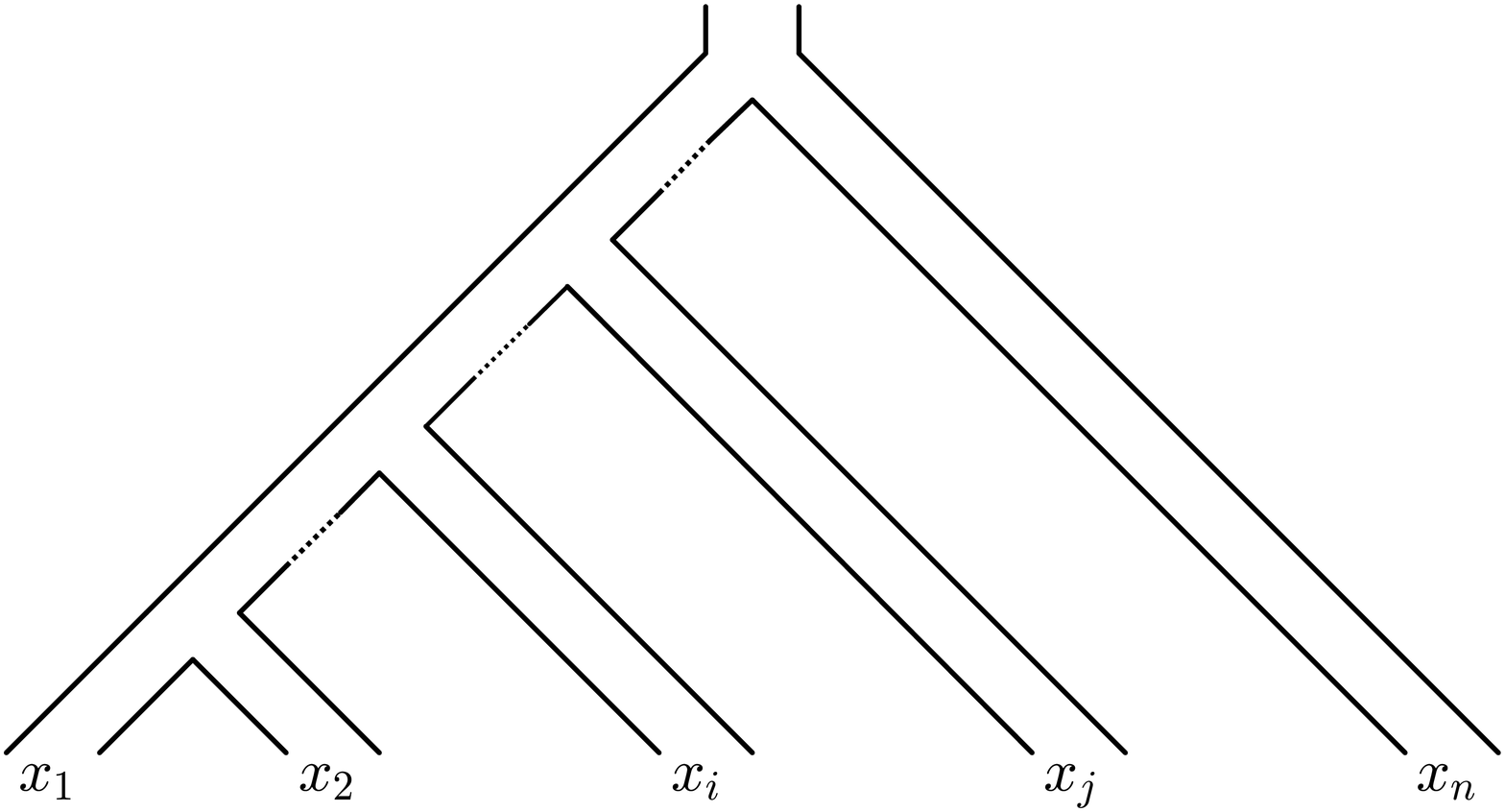}
    \includegraphics[scale=0.115]{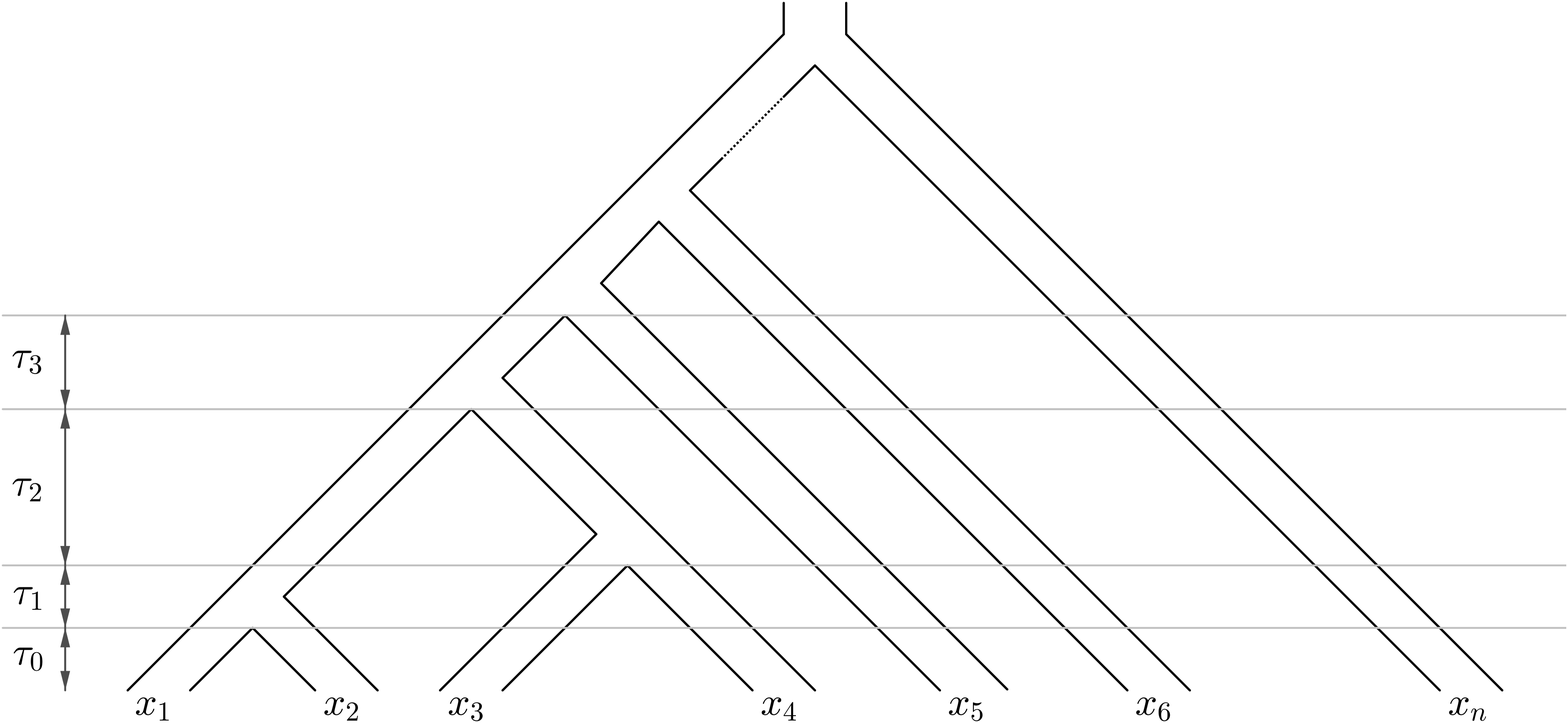}
    \caption{A caterpillar tree (top) and a pseudocaterpillar tree (bottom).}
    \label{Fig_CatPseudoCat}
\end{figure}

\paragraph*{The Fair Proportion index.}
Let $S$ be a rooted phylogenetic $X$-tree (not necessarily binary) with root $\rho$ and leaf set $X=\{x_1, \ldots, x_n\}$, where each edge $e$ is assigned a non-negative length $l(e) \in \mathbb{R}_{\geq 0}$. Then, the \emph{Fair Proportion (FP) index} \citep{Isaac2007,Redding2003} for leaf $x_i \in X$ is defined as 
\begin{align}
    FP_S(x_i) &= \sum\limits_{e \in P(S; \rho, x_i)} \frac{l(e)}{n(e)},
    \label{Def_FP}
\end{align}
where $P(S; \rho, x_i)$ denotes the path in $S$ from the root $\rho$ to leaf $x_i$ and $n(e)$ is the number of leaves descended from $e$. Note that $\sum\limits_{x_i \in X} FP_S(x_i) = \sum\limits_{e \in S} l(e),$
i.e., the FP index apportions the total sum of branch lengths of $S$ (also referred to as the `phylogenetic diversity' of $X$ \citep{Faith1992}) among the taxa in $X$. Note that we can calculate the FP index both for a given species tree $\cS$ as well as for all histories $h \in \cH(\cS)$ (where in case of the latter, the FP index is expressed in terms of both speciation and coalescent times). We denote these variants as $FP_{\cS}(x_i)$ and $FP_{h}(x_i)$, respectively. As an example, for the species tree $\cS=(S, \boldsymbol{\tau})$ and the history $h$ with its vector $\boldsymbol{t}_h$ of coalescent times shown in Figure \ref{Fig_5Leaves_Notation}, and the leaf $x_3$, we have
\begin{align*}
    FP_{\cS}(x_3) &= \frac{\tau_3}{3} + \frac{\tau_0+\tau_1+\tau_2}{1} = \tau_0 + \tau_1 + \tau_2 + \frac{\tau_3}{3} \\
    FP_{h}(x_3) &= \frac{t_4-t_3}{3} + \frac{\tau_3-t_2+t_3}{2} + \frac{\tau_0+\tau_1+\tau_2+t_2}{1} \\
    &= \tau_0 + \tau_1 + \tau_2 + \frac{\tau_3}{2} + \frac{t_2}{2} + \frac{t_3}{6} + \frac{t_4}{3}.
\end{align*}
Note that when the FP index is calculated for a particular history $h$ evolving within a species tree $\cS$, it will depend on some or all of the coalescent times $t_1, \ldots, t_{n-1}$. As there are infinitely many choices for the values $t_1, \ldots, t_{n-1}$ that satisfy the constraints of history $h$ evolving within $\cS$, in the following we define the expected FP index for leaf $x_i$ on history $h$ given the species tree $\cS$ by integrating over all possible times at which the coalescent events consistent with $h$ may occur. More precisely, we consider
\begin{align}
    \bE[FP_h(x_i)\vert\cS] &= \int FP_h(x_i) \, f_h(h, \boldsymbol{t}_h) \, d\boldsymbol{t}_h, \label{FP_history}
\end{align}
where $f_h(h, \boldsymbol{t}_h)$ denotes the probability density of history $h$ and $\boldsymbol{t}_h$ is the vector of coalescent times for $h$, and as before, the limits of integration are determined by the particular locations of the coalescent events of $h$ within the species tree. 
Based on this, we define the \emph{expected FP index of leaf $x_i$ on species tree $\cS$} as the expected FP index across all gene tree histories $h$ associated with $\cS$. More formally,
\begin{align}
    \bE[FP(x_i)\vert \cS] &= \sum\limits_{h \in \cH(\cS)}   \bE[FP_h(x_i)\vert\cS]
    =  \sum\limits_{h \in \cH(\cS)} \, \int FP_h(x_i) \, f_h(h, \boldsymbol{t}_h) \, d\boldsymbol{t}_h. \label{Expected_FP}
\end{align}

\paragraph*{Rankings.}
The FP index induces a natural ranking of the elements in $X$, where taxa are ordered according to their FP index. Here, a ranking $r$ is an assignment of ranking numbers to the elements of $X$, where for any pair of taxa $x_i, x_j \in X$, $x_i$ either receives a higher or lower ranking number than $x_j$, or the ranking numbers are equal (we then call $x_i$ and $x_j$ tied). We say that a function $f:X \rightarrow \mathbb{R}$ induces a ranking $r_f(X)$ if the ranking number of $x_i$ in $r_f$ is smaller than the ranking number of $x_j$ precisely if $f(x_i) > f(x_j)$. If $f(x_i)=f(x_j)$ for some $x_i \neq x_j$, $x_i$ and $x_j$ receive the same ranking number. In this case, we use the so-called `standard competition ranking' (also referred to as `1224' ranking), where tied elements receive the same ranking number, say $\widehat{r}$, and the subsequent ranking number is $\widetilde{r}=\widehat{r}+t-1$, where $t$ denotes the number of tied elements with ranking number $\widehat{r}$. In the following, we will consider the ranking of $X$ induced by the FP index on the species tree, denoted as $r_{(\cS,\text{FP})}(X)$, and the ranking of $X$ induced by the expected FP index, denoted as $r_{(\cS,\bE(\text{FP}))}(X)$, and analyze the relationship between these two rankings.

\section{Results} \label{Sec_Results}
We compare the rankings induced by the FP index on a given species tree $\cS$ with the ranking induced by the expected FP index on $\cS$. We show that under certain circumstances the rankings always coincide, while there also exist cases for which the rankings differ.

\subsection{Circumstances under which the FP index on the species tree and the expected FP index induce the same ranking}

In this section, we characterize certain tree topologies for which the rankings $r_{(\cS,\textup{FP})}(X)$ and $r_{(\cS,\bE(\textup{FP}))}(X)$ are guaranteed to coincide. 
However, we begin by showing that if a species tree $\cS$ on $X=\{x_1, \ldots, x_n\}$ contains a pendant subtree $\cS'$ on a subset $X' \subseteq X$, $X' = \{x_1, \ldots, x_m\}$ with $3 \leq m \leq n$ say, and such that the two maximal pendant subtrees of $\cS'$ consist of $m-1$ leaves labeled $x_1, \ldots, x_{m-1}$ and $1$ leaf labeled $x_m$, respectively, then among these $m$ leaves, leaf $x_m$ has both the highest FP index on $\cS$ as well as the highest expected FP index.

\begin{proposition} \label{Prop_ntaxon}
Let $n \geq 3$ and let $\cS=(S,\boldsymbol{\tau})$ be a species tree on $X=\{x_1, \ldots, x_n\}$ containing a pendant subtree $\cS'$ on $3 \leq m \leq n$ leaves, which in turn consists of a maximal pendant subtree on $m-1$ leaves, $x_1, \ldots, x_{m-1}$ say, and a maximal pendant subtree on one leaf, $x_m$ say, (Figure \ref{Fig_Prop_ntaxon}). Then,
\begin{enumerate}
\item[\textup{(i)}] $FP_{\cS}(x_m) > FP_{\cS}(x_i)$ for all $x_i \in \{x_1, \ldots, x_{m-1}\}$. 
\item[\textup{(ii)}] $\bE[FP(x_m)\vert\cS] > \bE[FP(x_i)\vert\cS]$ for all $x_i \in \{x_1, \ldots, x_{m-1}\}$. 
\end{enumerate}
\end{proposition}

\begin{figure}[htbp]
\centering
\includegraphics[scale=0.225]{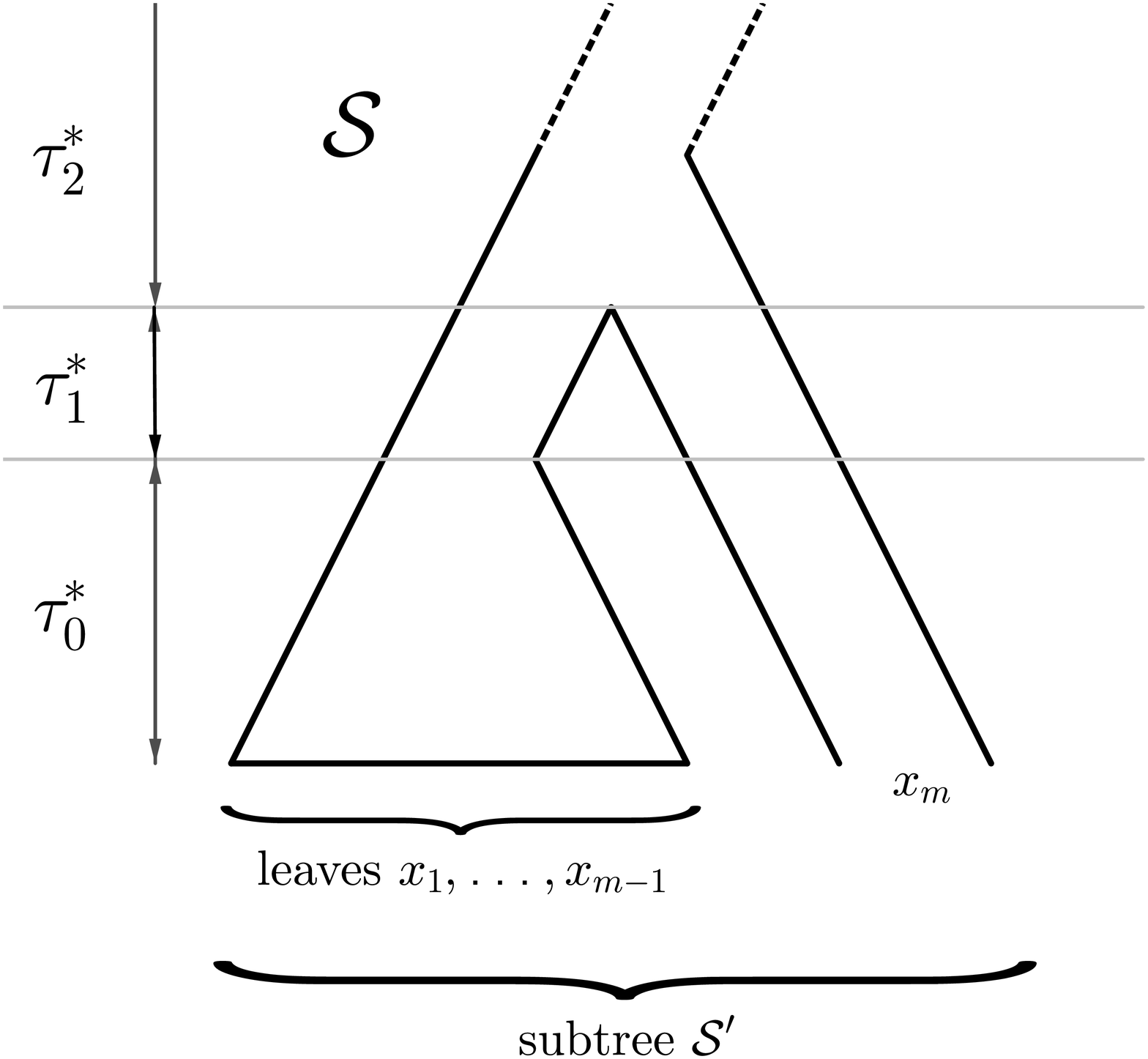}
\caption{Species tree $\cS$ as described in Proposition \ref{Prop_ntaxon}.}
\label{Fig_Prop_ntaxon}
\end{figure}

\begin{proof}\leavevmode 
\begin{enumerate}
\item[\textup{(i)}] 
First, note that all edges of $\cS$ above interval $\tau_1^\ast$ (see Figure \ref{Fig_Prop_ntaxon}) contribute the same amount, $\Delta_{FP}$ say, to $FP_{\cS}(x_m)$ and  $FP_{\cS}(x_i)$ for $x_i \in \{x_1, \ldots, x_{m-1}\}$.
In particular, for taxon $x_m$, we have $FP_{\cS}(x_m)= \tau_0^\ast + \tau_1^\ast + \Delta_{FP}$, whereas for $x_i \in \{x_1, \ldots, x_{m-1}\}$, we have $FP_{\cS}(x_i) \leq \tau_0^\ast+ \frac{\tau_1^\ast}{m-1} + \Delta_{FP} < FP_{\cS}(x_m)$ (since the subtree of $\cS'$ containing $x_i$ has $m-1 \geq 2$ leaves). This completes the first part of the proof.
 
\item[\textup{(ii)}] 
Let $x_i \in \{x_1, \ldots, x_{m-1}\}$. In order to show that $\bE[FP(x_m)\vert\cS] > \bE[FP(x_i)\vert\cS]$, we show that $$\bE[FP(x_m)-FP(x_i) \vert \cS] = \sum\limits_{h \in \cH(\cS)} \, \int \left[ FP_{h}(x_m) - FP_{h}(x_i) \right] f_h(h, \boldsymbol{t}_h) \, d\boldsymbol{t}_h  > 0.$$
Note that here (as everywhere in this manuscript) the limits of integration are determined by the particular locations of coalescent events of each history within the species tree.
We now analyze the set of histories $\mathcal{H}(\mathcal{S})$ more in-depth. In particular, we analyze them with respect to the interval $\tau_2^\ast$ in $\mathcal{S}$. Note that $x_i$ enters interval $\tau_2^\ast$ as part of some subtree, $T_{x_i}$ say, where the size of $T_{x_i}$ can range from 1 to $m-1$ leaves, whereas $x_m$ enters interval $\tau_2^\ast$ as a subtree of size 1, $T_{x_m}$ say.

However, subtrees $T_{x_i}$ and $T_{x_m}$ are exchangeable in the following sense: Each history $h \in \cH(\cS)$ either (a) contains $T_{x_i}$ and $T_{x_m}$ as a ``subtree cherry'' (i.e., the roots of $T_{x_i}$ and $T_{x_m}$ have a common parent and the two trees form a pendant subtree $(T_{x_i},T_{x_m})$ of $h$); or (b) if a history $h'$ does not contain $(T_{x_i},T_{x_m})$ as a pendant subtree, there exists a history $h''$ identical to $h'$ except that the roles of $T_{x_i}$ and $T_{x_m}$ are exchanged and such that $\boldsymbol{t}_{h'} = \boldsymbol{t}_{h''}$ and $f_{h'}(h',\boldsymbol{t}_{h'})=f_{h''}(h'',\boldsymbol{t}_{h''})$ (Figure \ref{Fig_Prop1}).
In both cases, let $lca(x_i,x_m)$ denote the lowest common ancestor of $x_i$ and $x_m$ in the corresponding gene tree history, and consider the FP indices of $x_i$ and $x_m$.

\begin{enumerate}[(a)]
    \item If $h$ is such that $T_{x_i}$ and $T_{x_m}$ form a pendant subtree of $h$ with root $lca(x_i,x_m)$, all edges of $h$ above $lca(x_i,x_m)$ (i.e., all edges in the path from the root of $h$ to $lca(x_i,x_m)$) contribute the same amount to $FP_h(x_i)$ and $FP_h(x_m)$. However, considering the paths from $lca(x_i,x_m)$ to $x_i$ and $x_m$ of length $\tau_0^\ast+\tau_1^\ast+t_{lca}$, where $t_{lca}$ denotes the time of the coalescence event that leads to $lca(x_i,x_m)$, the contribution of edges in the path from $lca(x_i,x_m)$ to $x_m$ to $FP_h(x_m)$ is precisely $\tau_0^\ast+\tau_1^\ast+t_{lca}$, whereas the contribution of edges in the path from $lca(x_i,x_m)$ to $x_i$ to $FP_h(x_i)$ is less or equal to $\tau_0^\ast+\tau_1^\ast+t_{lca}$ (where the inequality is strict if $T_{x_i}$ contains at least two leaves). In total, this implies that we have $FP_h(x_m) \geq FP_h(x_i)$.
    \item Now, if $h'$ and $h''$ are two histories that differ only by a permutation of the roles of $T_{x_i}$ and $T_{x_m}$, we first note that the lowest common ancestor of $x_i$ and $x_m$ corresponds to the same vertex $lca(x_i,x_m)$ in $h'$ and $h''$ (see Figure \ref{Fig_Prop1}(b)). Moreover, as in case (a), all edges above $lca(x_i,x_m)$ contribute the same amount towards $FP_{h'}(x_m)$, $FP_{h'}(x_i)$, $FP_{h''}(x_m)$, and $FP_{h''}(x_i)$. 
    Now, as $T_{x_i}$ contains at least one leaf, whereas $T_{x_m}$ contains precisely one leaf, referring to the notation in Figure \ref{Fig_Prop1}, we can conclude that the contribution of edges in the path from $lca(x_i,x_m)$ to $u$ in $h''$ towards $FP_{h''}(x_m)$ is greater than or equal to the contribution of the corresponding edges in the path from $lca(x_i,x_m)$ to $u$ in $h'$ towards $FP_{h'}(x_i)$. Moreover, similar to case (a), the contribution of edges in the path from $u$ to $x_m$ in $h''$ towards $FP_{h''}(x_m)$ is greater than or equal to the contribution of edges in the path from $u$ to $x_i$ in $h'$ towards $FP_{h'}(x_i)$. In summary, $FP_{h''}(x_m) \geq FP_{h'}(x_i)$. An analogous argument shows that $FP_{h'}(x_m) \geq FP_{h''}(x_i)$. In particular note that the inequalities are strict if $T_{x_i}$ contains at least two leaves.
\end{enumerate}

Thus, we can partition the set of histories $\cH(\cS)$ into the subset of histories of type (a), in which $T_{x_i}$ and $T_{x_m}$ form a pendant subtree and we have $FP_h(x_m) \geq FP_h(x_i)$, and thus
$$ \int \underbrace{\left( FP_h(x_m)- FP_h(x_i) \right)}_{\geq 0} \, f_h(h, \boldsymbol{t}_h) \, d\boldsymbol{t}_{h} \geq  0,$$
and into the subset of histories of type (b), i.e., into pairs of ``matching'' histories $(h',h'')$, where, with the reasoning from above, in particular, recalling that $\boldsymbol{t}_{h'}=\boldsymbol{t}_{h''}$ and $f_{h'}(h',\boldsymbol{t}_{h'}) = f_{h''}(h'',\boldsymbol{t}_{h''})$, we have for each such pair $(h',h'')$ that
$$ \int\limits \left[ \left(\underbrace{FP_{h'}(x_m)-FP_{h''}(x_i)}_{\geq 0}\right)+   \left(\underbrace{FP_{h''}(x_m)-FP_{h'}(x_i)}_{\geq 0}\right) \right] f_{h'}(h',\boldsymbol{t}_{h'}) d\boldsymbol{t}_{h'}  \geq 0.$$

Now, as in both cases the inequalities are strict for all histories $h$ such that $T_{x_i}$ contains at least two leaves (and such histories always exist), in summary, we have
$$\bE[FP(x_m)-FP(x_i) \vert \cS] = \sum\limits_{h \in \cH(\cS)} \, \int  \left[ FP_{h}(x_m) - FP_{h}(x_i) \right] f_h(h, \boldsymbol{t}_h) \, d\boldsymbol{t}_h  > 0$$
as claimed. This completes the proof.
\end{enumerate}
\end{proof}

\begin{figure}[htbp]
     \centering
     \begin{subfigure}{1\textwidth}
         \centering
         \includegraphics[scale=0.175]{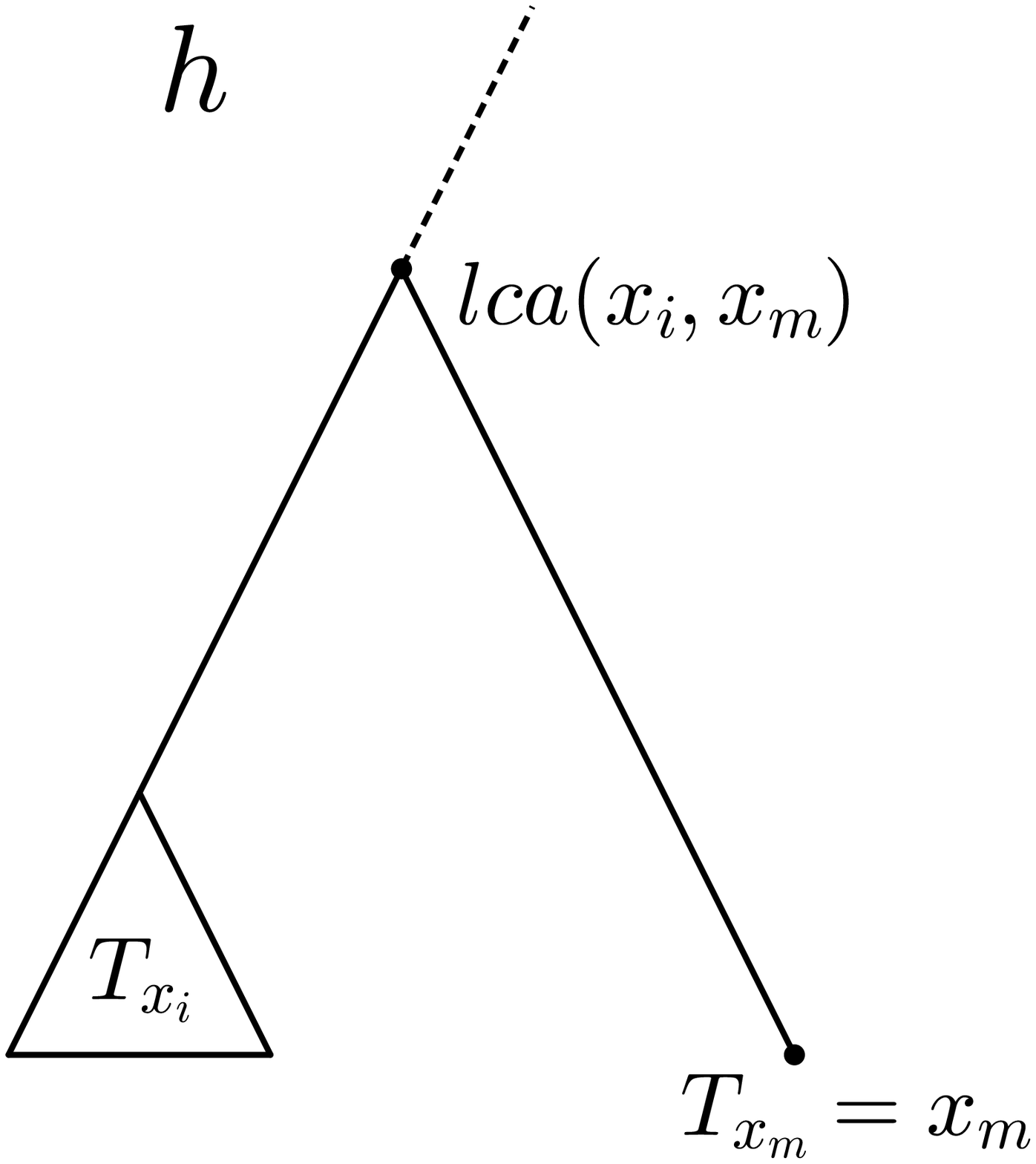}
         \caption{\phantom{.}}
     \end{subfigure}
     \newline
     \begin{subfigure}{1\textwidth}
         \centering
         \includegraphics[scale=0.175]{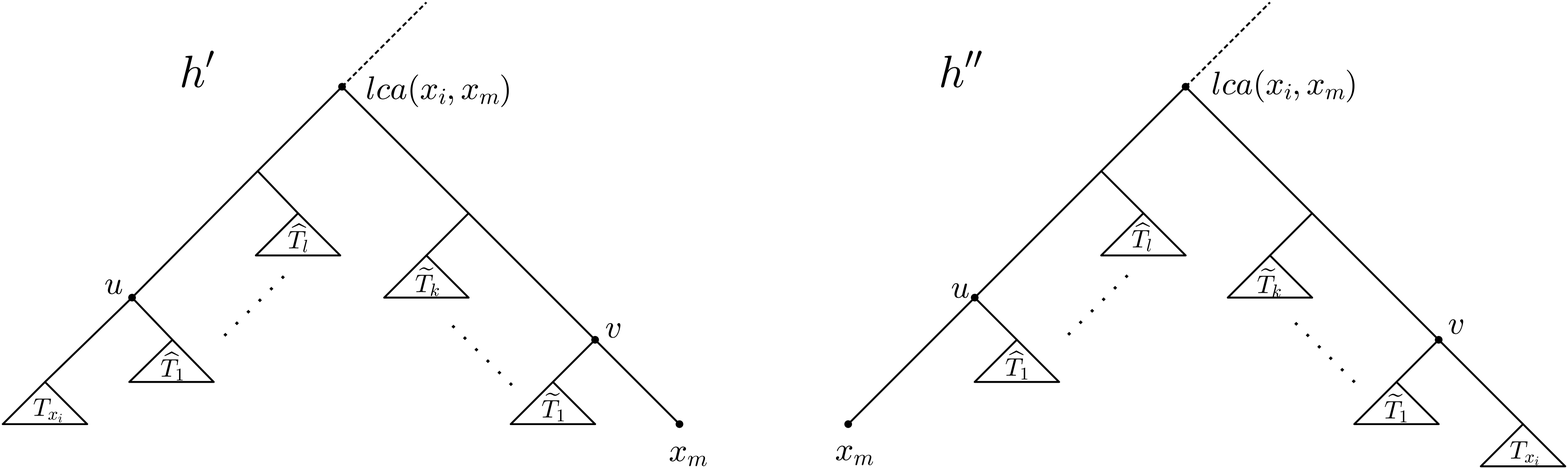}
         \caption{\phantom{.}}
     \end{subfigure}
      \caption{Histories in $\cH(\cS)$ of type (a) and (b) as described in the proof of Proposition \ref{Prop_ntaxon}.}
    \label{Fig_Prop1}
\end{figure}

We now show that the ranking induced by the FP index on a species tree $\cS$ coincides with the ranking induced by the expected FP index for $\cS$ if $\cS$ is a caterpillar or pseudocaterpillar tree.

\begin{theorem}\label{Thm_Caterpillar_Pseudocaterpillar}
Let $n \geq 2$ and let $\cS=(S,\boldsymbol{\tau})$ be a caterpillar species tree on $X=\{x_1, \ldots, x_n\}$. Then, if $FP_{\cS}(x_j) > FP_{\cS}(x_i)$ for $x_i, x_j \in X$, we also have $\bE[FP(x_j)\vert \cS] > \bE[FP(x_i) \vert \cS]$. Moreover, if $FP_{\cS}(x_j) = FP_{\cS}(x_i)$, then $\bE[FP(x_j) \vert \cS] = \bE[FP(x_i) \vert \cS]$. In particular, the rankings $r_{(\cS,\textup{FP})}(X)$ and $r_{(\cS,\bE(\textup{FP}))}(X)$ coincide. The same statements hold if $n \geq 4$ and $\cS=(S,\boldsymbol{\tau})$ is a pseudocaterpillar species tree.
\end{theorem}

\begin{proof}
First, assume that $\cS$ is a caterpillar species tree and without loss of generality let $\cS$ be as depicted in Figure \ref{Fig_CatPseudoCat}. Then, we clearly have $FP_{\cS}(x_n) > FP_{\cS}(x_{n-1}) > \ldots > FP_{\cS}(x_3) > FP_{\cS}(x_2) = FP_{\cS}(x_1)$.
Let $x_i, x_j \in X$ with $FP_{\cS}(x_j) > FP_{\cS}(x_i)$.
Then it follows immediately from Proposition \ref{Prop_ntaxon} (with $x_j=x_m$ and $x_i \in \{x_1, \ldots, x_{m-1}\}$) that $\bE[FP(x_j) \vert \cS] > \bE[FP(x_i) \vert \cS]$.

Now, if $FP_{\cS}(x_j)=FP_{\cS}(x_i)$, we must have $\{x_i,x_j\} = \{x_1,x_2\}$. In particular, $x_i$ and $x_j$ form the unique cherry in $\cS$. This immediately implies that $x_i$ and $x_j$ are exchangeable and thus $\bE[FP(x_j) \vert \cS] = \bE[FP(x_i) \vert \cS]$ as claimed. 

As $FP_{\cS}(x_n) > FP_{\cS}(x_{n-1}) > \ldots > FP_{\cS}(x_3) > FP_{\cS}(x_2) = FP_{\cS}(x_1)$ and all inequality and equality signs are ``preserved'' when considering the expected FP index, we clearly have $r_{(\cS,\textup{FP})}(X)=r_{(\cS,\bE(\textup{FP}))}(X)$. This completes the first part of the proof.\\

\noindent Now, assume that $\cS$ is a pseudocaterpillar species tree with $n\geq 4$ leaves as depicted in Figure \ref{Fig_CatPseudoCat}. In this case, we have $FP_{\cS}(x_n) > FP_{\cS}(x_{n-1}) > \ldots > FP_{\cS}(x_5) > FP_{\cS}(x_4) = FP_{\cS}(x_3) > FP_{\cS}(x_2) = FP_{\cS}(x_1)$. In particular, if $FP_{\cS}(x_i) = FP_{\cS}(x_j)$, we have $\{x_i,x_j\} = \{x_1,x_2\}$ or $\{x_i,x_j\} = \{x_3,x_4\}$. In both cases, $x_i$ and $x_j$ form a cherry and are clearly exchangeable. We thus have $\bE[FP(x_j) \vert \cS] = \bE[FP(x_i) \vert \cS]$ as claimed. Now, consider $x_i, x_j \in X$ with $FP_{\cS}(x_j) > FP_{\cS}(x_i)$. If $x_j \in \{x_5, \ldots, x_n\}$ (and $x_i \in \{x_1, \ldots, x_n\}$ with $i < j$) it follows from Proposition \ref{Prop_ntaxon} that $\bE[FP(x_j) \vert \cS] > \bE[FP(x_i) \vert \cS]$ as claimed. It thus only remains to consider the case that $x_j \in \{x_3,x_4\}$ and $x_i \in \{x_1,x_2\}$. Without loss of generality, assume $x_j=x_3$ and $x_i=x_1$. Considering 
\begin{align*}
\bE[FP(x_3)-FP(x_1) \vert \cS] 
&= \sum\limits_{h \in \cH(\cS)} \, \int \left[ FP_{h}(x_3) - FP_{h}(x_1) \right] f_h(h, \boldsymbol{t}_h) \, d\boldsymbol{t}_h, 
\end{align*}
we now note that $\cH(\cS)$ can be partitioned into five disjoint subsets:
\begin{enumerate}[(i)]
    \item The subset of histories, $\cH_1(\cS)$ say, for which both the coalescence involving $x_1$ and the coalescence involving $x_3$ occur in interval $\tau_3$ or later. In this case, $x_1$ and $x_3$ are exchangeable and we can conclude that
    \begin{align*}
       \sum\limits_{h \in \cH_1(\cS)} \, \int \left[ FP_{h}(x_3) - FP_{h}(x_1) \right] f_h(h, \boldsymbol{t}_h) \, d\boldsymbol{t}_h = 0.
    \end{align*}
    \item The subset of histories, $\cH_2(\cS)$ say, for which both the coalescence involving $x_1$ and the coalescence involving $x_3$ occur in interval $\tau_2$. By the symmetry of the 4-leaf subtree of $\cS$ containing leaves $\{x_1,x_2,x_3,x_4\}$, leaves $x_1$ and $x_3$ are also exchangeable (for each history $h'$ where the cherry $[x_1,x_2]$ is formed prior to the cherry $[x_3,x_4]$, there exists a history $h''$ identical to $h'$ except that the cherry $[x_3,x_4]$ is formed prior to the cherry $[x_1,x_2]$; in particular, the roles of $x_1$ and $x_3$ are interchanged between $h'$ and $h''$) and we can conclude that 
    \begin{align*}
        \sum\limits_{h \in \cH_2(\cS)} \, \int \left[ FP_{h}(x_3) - FP_{h}(x_1) \right] f_h(h, \boldsymbol{t}_h) \, d\boldsymbol{t}_h = 0.
    \end{align*}
    \item The subset of histories, say $\cH_3(\cS)$, for which the coalescence involving $x_1$ occurs in interval $\tau_1$ and the coalescence involving $x_3$ occurs in interval $\tau_2$ or later. In this case, $x_1$ enters the interval $\tau_3$ as part of a subtree, $T_{x_1}$ say, containing precisely two leaves, whereas $x_3$ enters the interval $\tau_3$ as part of a subtree, $T_{x_3}$ say, containing at most two leaves. 
    As in the proof of Proposition \ref{Prop_ntaxon}, $T_{x_1}$ and $T_{x_3}$ are exchangeable in the sense that each history $h \in \cH_3(\cS)$ either (a) contains a pendant subtree $(T_{x_1}, T_{x_3})$; or (b) if a history $h' \in \cH_3(\cS)$ does not contain $(T_{x_1},T_{x_3})$ as a pendant subtree, there exists a history $h''$ identical to $h'$ except that the roles of $T_{x_1}$ and $T_{x_3}$ are exchanged and such that $\boldsymbol{t}_{h'} = \boldsymbol{t}_{h''}$ and $f_{h'}(h',\boldsymbol{t}_{h'})=f_{h''}(h'',\boldsymbol{t}_{h''})$. Moreover, as $x_1$ coalesces strictly before $x_3$, for each history in $\cH_3(\cS)$ the pendant edge incident to $x_3$ is strictly larger than the pendant edge incident to $x_1$. 
    Using this and the exchangeability of $T_{x_1}$ and $T_{x_3}$ together with the fact that $T_{x_3}$ contains at most two leaves, whereas $T_{x_1}$ contains precisely two leaves, we can thus conclude that 
    \begin{align*}
        \sum\limits_{h \in \cH_3(\cS)} \, \int \left[ FP_{h}(x_3) - FP_{h}(x_1) \right] f_h(h, \boldsymbol{t}_h) \, d\boldsymbol{t}_h > 0.
    \end{align*}
    \end{enumerate}
The last two remaining subsets are the subsets, $\mathcal{H}_4(\cS)$ and $\mathcal{H}_5(\cS)$ say, where $\mathcal{H}_4(\cS)$ is the subset of histories for which the coalescence involving $x_1$ occurs in $\tau_2$ and the coalescence involving $x_3$ occurs in $\tau_3$ or later, and $\mathcal{H}_5(\cS)$ is the subset of histories for which the coalescence involving $x_1$ occurs in $\tau_3$ or later and the coalescence involving $x_3$ occurs in $\tau_2$. By the symmetry of the 4-leaf subtree of $\cS$ containing leaves $\{x_1, x_2, x_3, x_4\}$, the sets $\mathcal{H}_4(\cS)$ and $\mathcal{H}_5(\cS)$ have the same cardinality. Now, for each history $h' \in \mathcal{H}_4(\mathcal{S})$, there exists precisely one history $h'' \in \mathcal{H}_5(\mathcal{S})$ identical to $h'$, except that in $h'$ the most recent coalescence at time $t_1$ corresponds to the formation of the cherry $[x_1,x_2]$, whereas in $h''$, the most recent coalescence at time $t_1$ corresponds to the formation of the cherry $[x_3,x_4]$; in particular, the roles of $x_1$ and $x_3$ are interchanged between $h'$ and $h''$ (similar to Case 2 discussed above, except that not necessarily both cherries $[x_1,x_2]$ and $[x_3,x_4]$ are formed in $h'$, respectively $h''$). Moreover, $\boldsymbol{t}_{h'}=\boldsymbol{t}_{h''}$ and $f_{h'}(h',\boldsymbol{t}_h')=f_{h''}(h'',\boldsymbol{t}_h'')$. Furthermore, $FP_{h'}(x_1)=FP_{h''}(x_3)$ and $FP_{h''}(x_1)=FP_{h'}(x_3)$, and thus we can cancel the contribution of $h'$ towards $\bE[FP(x_3)-FP(x_1) \vert \cS]$ with the corresponding contribution of $h''$.   
    
Taking all cases together, this implies that $\bE[FP(x_3)-FP(x_1) \vert \cS] > 0$ and thus $\bE[FP(x_3) \vert \cS] > \bE[FP(x_1) \vert \cS]$ as claimed. 
Now, as $FP_{\cS}(x_n) > FP_{\cS}(x_{n-1}) > \ldots > FP_{\cS}(x_4) = FP_{\cS}(x_3) > FP_{\cS}(x_2) = FP_{\cS}(x_1)$ and all inequality and equality signs are preserved when considering the expected FP index, we clearly have $r_{(\cS,\textup{FP})}(X)=r_{(\cS,\bE(\textup{FP}))}(X)$. This completes the proof.
\end{proof}

\noindent As all species trees with at most four leaves are either caterpillar or pseudocaterpillar trees, Theorem \ref{Thm_Caterpillar_Pseudocaterpillar} immediately leads to the following corollary:

\begin{corollary}\label{Cor_atmost4}
If $\cS=(S,\boldsymbol{\tau})$ is a species tree on $X = \{x_1, \ldots, x_n\}$ with $n \leq 4$, the rankings $r_{(\cS,\textup{FP})}(X)$ and $r_{(\cS,\bE(\textup{FP}))}(X)$ coincide.
\end{corollary}

\subsection{On the existence of species trees with different rankings induced by the FP index and the expected FP index for all \texorpdfstring{$n \geq 5$}{n >= 5}}

After having analyzed some circumstances under which the FP index calculated for a given species tree and the expected FP index across all gene tree histories associated with this species tree induce the same ranking, in this section we show that for all leaf numbers greater or equal to 5, there exist species trees for which the two rankings differ.

\begin{theorem}\label{Thm_rankswap}
For all $n \in \mathbb{N}_{\geq 5}$, there exists a species tree $\cS=(S,\boldsymbol{\tau})$ on $X = \{x_1, \ldots, x_n\}$ such that $\cS$ contains two leaves, $x_i$ and $x_j$ say, with $FP_{\cS}(x_i) > FP_{\cS}(x_j)$ but $\bE[FP(x_i) \vert \cS] < \bE[FP(x_j) \vert \cS]$. In particular, $\cS $ is such that the rankings $r_{(\cS, \textup{FP})}(X)$ and $r_{(\cS, \bE(\textup{FP}))}(X)$ are different.
\end{theorem}

\noindent The strategy to prove this theorem is to show that there exists a species tree $\cS=(S,\boldsymbol{\tau})$ on precisely 5 leaves for which the rankings $r_{(\cS, \textup{FP})}(X)$ and $r_{(\cS, \bE(\textup{FP}))}(X)$ differ. We then construct a species tree $\cS'$ on $n > 5$ leaves containing $\cS$ as a subtree and show that the rank swap in $\cS$ is preserved in $\cS'$. For this purpose, we will decompose all histories evolving within $\cS$, respectively $\cS'$, by ``cutting'' them at a certain height determined by the species tree. We now explain this procedure in more detail.\\

\noindent Let $\cS=(S,\boldsymbol{\tau})$ be a species tree of height $\eta$ and let $h \in \cH(\cS)$ be a history evolving within $\cS$. Let $0< \eta' \leq \eta$. By \emph{cutting} history $h$ at height $\eta'$, we mean subdividing all edges of $h$ present at time $\eta'$ into two new edges, thereby disconnecting $h$. In case there is an interior vertex $v$ of $h$ present at time $\eta'$, we also subdivide it into two copies, one of which forming the root of a subtree of $h$ with the outgoing edges, and the other being a new leaf for the incoming edge. This procedure will subdivide $h$ into a forest of several trees, one of which, say $h^\mathit{top}$, will contain the root of $h$. We call $h^\mathit{top}$ the \emph{top part of $h$ concerning cut height $\eta'$} and we refer to the other trees, say $h^1, \ldots, h^k$, in the forest as the \emph{chopped parts of $h$ concerning cut height $\eta'$}. Moreover, if we attach to each leaf of $h^\mathit{top}$ the number of descending leaves it had in $h$ by edges of length 0 (note that the resulting tree need not be binary and may contain degree-2 vertices), we derive the \emph{extended top part of $h$ concerning cut height $\eta'$}, which we denote by $h^\mathit{top}_\mathit{ext}$. Whenever there is no ambiguity concerning the cut height, we drop the explicit mention of $\eta'$. 
As an example, consider the species tree $\cS$ and the history $h$ evolving within $\cS$ depicted in Figure~\ref{Fig_5Leaves_Notation} and let $\eta'=\tau_0+\tau_1+\tau_2$. Then cutting $h$ at height $\eta'$ leads to the decomposition of $h$ into its (extended) top and chopped parts as depicted in Figure~\ref{Fig_Cutting}.

\begin{figure}[htbp]
    \centering
    \includegraphics[scale=0.2]{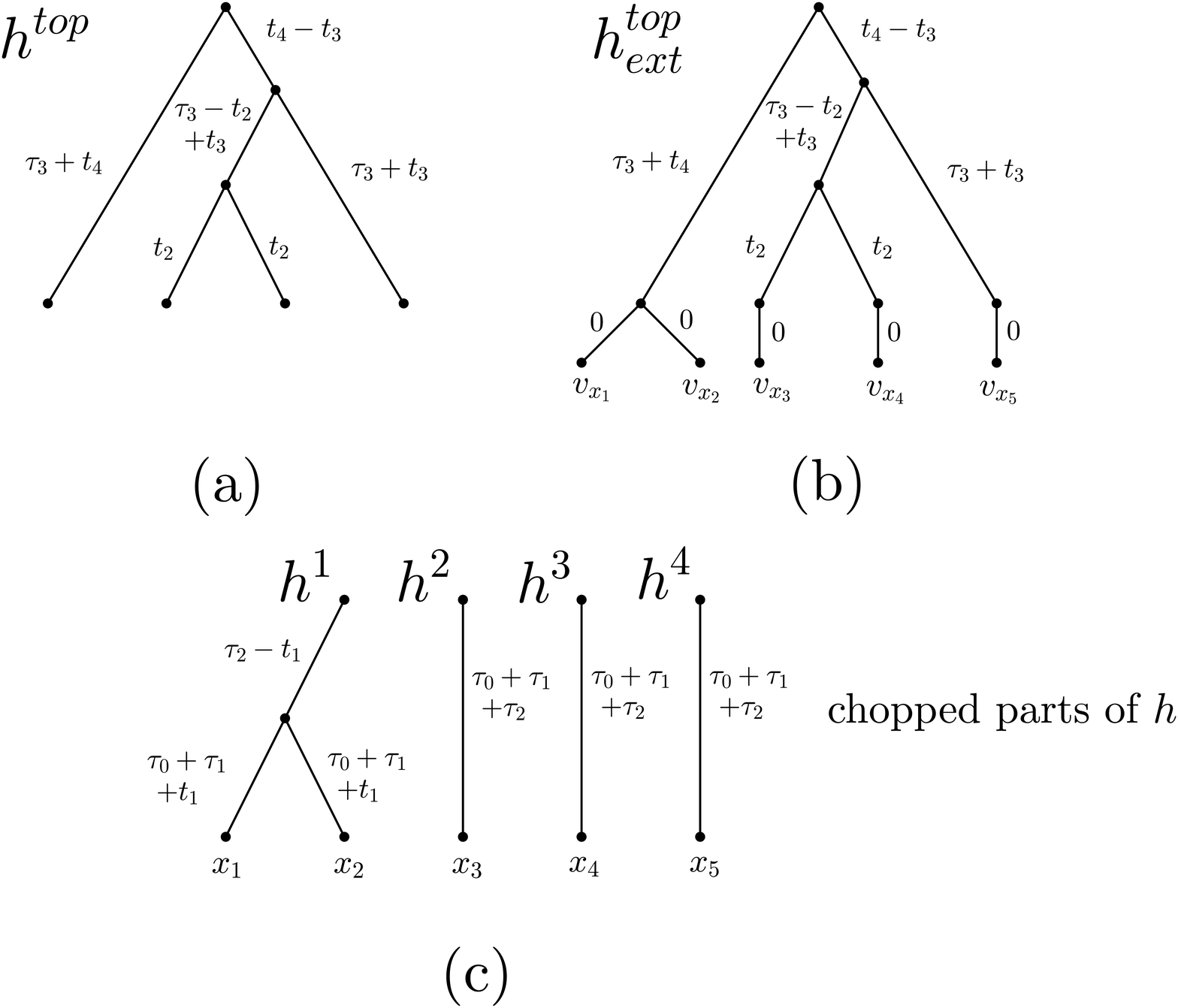}
    \caption{Decomposition of gene tree history $h$ evolving within species tree $\cS$ as depicted in Figure~\ref{Fig_5Leaves_Notation} when cutting at height $\eta'=\tau_0+\tau_1+\tau_2$ into its (a) top part, (b) extended top part, and (c) chopped parts.}
    \label{Fig_Cutting}
\end{figure}

\noindent We now show that given a cut height $0 < \eta' \leq \eta$, the FP index of a taxon $x_i$ on history $h$ can be calculated by considering the decomposition of $h$ into its extended top part and chopped parts concerning cut height $\eta'$.

\begin{lemma}\label{Lemma_cutting}
Let $\cS=(S,\boldsymbol{\tau})$ be a species tree on taxon set $X=\{x_1, \ldots, x_n\}$ and let $\eta=\tau_0 + \ldots + \tau_{n-2}$ denote its height. Let $h \in \cH(\cS)$ be a history evolving within $\cS$. Moreover, let $0< \eta' \leq \eta$, and let $h^\mathit{top}_\mathit{ext}$ be the extended top part of $h$ concerning cut height $\eta'$ and let $h^1, \ldots, h^k$ be the chopped parts of $h$. Then, for each $x_i$ in $X$, we have
\begin{align}
    FP_h(x_i) = FP_{h^i}(x_i) + FP_{h^\mathit{top}_\mathit{ext}}(v_{x_i}), \label{FP_Cut}
\end{align}
where $h^i$ is the chopped part of $h$ containing $x_i$ (note that the root of $h^i$ might have out-degree 1) and $v_{x_i}$ is one of the newly attached leaves of $h^\mathit{top}_\mathit{ext}$ whose parent is an ancestor of $x_i$ in $h$.
\end{lemma}

\begin{proof}
Considering the path $P(h; \rho, x_i)$ from the root $\rho$ in $h$ to leaf $i$ with respect to cut height $\eta'$, we can decompose this path into two subpaths, $P_i$ and $P_{top}$ say, where $P_i$ is subpath belonging to $h^i$, except possibly for a subdivided edge $e_{cut}$, and $P_{top}$ is the subpath belonging to the top part $h^\mathit{top}$ of history $h$, again except for edge $e_{cut}$ (if it exists).
In case such a subdivided edge $e_{cut}$ exists, let us denote by $e_i$ the part of $e_{cut}$ that belongs to the chopped part $h^i$ of history $h$ and by $e_{top}$ the part of $e_{cut}$ which belongs to the top part $h^\mathit{top}$ of history $h$. Setting $l(e_{cut})=0$ (implying $l(e_i)=l(e_{top})=0$) in case no edge in $P(h;\rho,x_i)$ was subdivided and noting that $n(e_{cut})=n(e_i)=n(e_{top})$, we get
\begin{align*}
    FP_h(x_i) &= \sum_{e \in P(h;\rho,x_i)} \frac{l(e)}{n(e)} 
    = \sum_{e \in P_i} \frac{l(e)}{n(e)} + \sum_{e \in P_{top}} \frac{l(e)}{n(e)} + \frac{l(e_i)}{n(e_i)} + \frac{l(e_{top})}{n(e_{top})} \\
    &= \underbrace{\left( \sum_{e \in P_i} \frac{l(e)}{n(e)} + \frac{l(e_i)}{n(e_i)} \right)}_{= FP_{h^i}(x_i)} + \underbrace{\left(\sum_{e \in P_{top}} \frac{l(e)}{n(e)} + \frac{l(e_{top})}{n(e_{top})}\right)}_{= FP_{h^\mathit{top}_\mathit{ext}}(v_{x_i})} \\
    &= FP_{h^i}(x_i) + FP_{h^\mathit{top}_\mathit{ext}}(v_{x_i}).
\end{align*}
This completes the proof.
\end{proof}

\noindent As an example consider history $h$ depicted in Figure~\ref{Fig_5Leaves_Notation} and its decomposition concerning cut height $\eta'=\tau_0+\tau_1+\tau_2$ depicted in Figure~\ref{Fig_Cutting}. Here, we have $FP_h(x_1) = \tau_0 + \tau_1 + \frac{\tau_2+\tau_3+t_1+t_4}{2}$, $FP_{h^1}(x_1)=\tau_0+\tau_1+\frac{\tau_2+t_1}{2}$, and $FP_{h^\mathit{top}_\mathit{ext}}(v_{x_1}) = \frac{\tau_3+t_4}{2}$. In particular, $FP_h(x_1) = FP_{h^1}(x_1) + FP_{h^\mathit{top}_\mathit{ext}}(v_{x_1})$ as stated by Lemma~\ref{Lemma_cutting}. \\

\noindent Now, considering the expected FP index across all gene tree histories associated with a species tree $\cS$ and given a cut height $0 < \eta' \leq \eta$, as a consequence of Lemma~\ref{Lemma_cutting} and using linearity of the expectation, we obtain the following decomposition of Equation~\eqref{Expected_FP}:
\begin{align}
\bE[FP(x_i)\vert \cS] 
&= \sum\limits_{h \in \cH(\cS)} \bE[FP_h(x_i)\vert\cS] \nonumber \\
&= \sum\limits_{h \in \cH(\cS)} \bE[FP_{h^i}(x_i) + FP_{h^\mathit{top}_\mathit{ext}}(v_{x_i}) \vert \cS] \nonumber \\
&=  \sum\limits_{h \in \cH(\cS)} \left( \bE[FP_{h^i}(x_i) \vert \cS] + \bE[FP_{h^\mathit{top}_\mathit{ext}}(v_{x_i})\vert \cS] \right)\nonumber \\
&= \underbrace{\sum\limits_{h \in \cH(\cS)} \bE[FP_{h^i}(x_i) \vert \cS]}_{\eqqcolon \bE[FP_\text{chop}^{\eta'}(x_i)\vert \cS]}
+ \underbrace{\sum\limits_{h \in \cH(\cS)} \bE[FP_{h^\mathit{top}_\mathit{ext}}(v_{x_i})\vert \cS]}_{\eqqcolon \bE[FP_\text{top}^{\eta'}(x_i)\vert\cS]}.  \label{FP_Chop_Top_Decomp}
\end{align}
In the following, we will refer to the first summand in Equation~\eqref{FP_Chop_Top_Decomp} as the `expected chopped FP index' of $x_i$, denoted as $\bE[FP_\text{chop}^{\eta'}(x_i)\vert \cS]$, and to the second summand as the `expected top FP index' of $x_i$, denoted as $\bE[FP_\text{top}^{\eta'}(x_i)\vert\cS]$.
With this we are now finally in the position to prove Theorem~\ref{Thm_rankswap}.

\begin{proof}[Proof of Theorem~\ref{Thm_rankswap}]
We first show that the statement holds for $n=5$. Let $\cS=(S,\boldsymbol{\tau})$ be as depicted in Figure~\ref{Fig_5Leaves_Notation} and set $\boldsymbol{\tau}=(\tau_0,\tau_1,\tau_2,\tau_3) = (2.114, 0.1269, 0.00293, 0.37684)$. Then, it can easily be seen that $FP_{\cS}(x_3) \approx 2.36944 > 2.367335 = FP_{\cS}(x_1)$. Moreover, we used the computer algebra system Mathematica \citep{Mathematica} to explicitly enumerate all gene tree histories in $\cH(\cS)$ and calculate the expected FP indices of all taxa according to Equation \eqref{Expected_FP} (the full expressions in terms of $(\tau_0,\tau_1,\tau_2,\tau_3,\tau_4)$ are given in the Appendix). In particular,  $\bE[FP(x_3)\vert \cS] \approx 3.23846178 < 3.2940675 \approx \bE[FP(x_1)\vert \cS] $. This completes the proof for $n=5$. 

In the following, we will construct a species tree $\cS'$ on $n > 5$ leaves containing $\cS$ as a pendant subtree such that the rank swap of taxa $x_1$ and $x_3$ is preserved in $\cS'$. However, before we describe this construction in more detail, we remark that if we cut all histories in $\cH(\cS)$ at the height of $\cS$, i.e., at height $\eta = \tau_0 + \tau_1 + \tau_2 + \tau_3$, using Mathematica \citep{Mathematica} we obtained the following expected chopped FP indices for taxa $x_1$ and $x_3$ and the choice of $\boldsymbol{\tau}$ given above:
\begin{align*}
    \bE[FP_{\text{chop}}^\eta(x_1)\vert \cS] &\approx 2.5661 \quad \text{and} \quad 
    \bE[FP_{\text{chop}}^\eta(x_3)\vert \cS] \approx 2.5641.
\end{align*}
In particular, the order of $x_1$ and $x_3$ induced by the expected chopped FP index coincides with the order induced by the overall expected FP index and contradicts the order induced by the FP index on the species tree. Moreover, considering the extended top parts of all histories $h \in \cH(\cS)$ with respect to cut height $\eta$, using Mathematica \citep{Mathematica} we calculated the distribution of the number of new leaves that are attached to the last ancestor of $x_1$, respectively $x_3$, in the top part $h^\mathit{top}$ of each history $h$ (again for the choice of $\boldsymbol{\tau}$ given above). Let $n_1$ be the number of leaves attached to the last ancestor of $x_1$ and let $n_3$ denote the number of leaves attached to the last ancestor of $x_3$. Then, $n_1 \in \{1,2\}$ (if a history $h$ is such that $x_1$ and $x_2$ coalesced before time $\eta$, $n_1$ will be two; otherwise $n_1$ will be one) and $n_3 \in \{1,2,3\}$ (where $n_3=3$ if $h$ is such that all of $x_3,x_4,x_5$ coalesced before time $\eta$, $n_3=2$ if $x_3$ coalesced with $x_4$ or $x_5$ but all three taxa did not coalesce, and $n_3=1$ if $x_3$ did not coalesce) and we have the following distribution (Table~\ref{Table_ni_dist}):
    \begin{table}[htbp]
    \caption{Distribution of $n_1$ and $n_3$ as described in the proof of Theorem \ref{Thm_rankswap}.}
    \centering
    \begin{tabular}{|c|c|c|c|}
    \hline
    $m$ & 1 & 2 & 3\\
    \hline
    $\bP(n_1=m)$ & 0.602499 & 0.397501 & 0 \\
    \hline
    $\bP(n_3=m)$ & 0.504977 & 0.362098 & 0.132925 \\
    \hline
    \end{tabular} 
    \label{Table_ni_dist}
    \end{table}
    
Keeping these findings in mind, we now construct a species tree $\cS'=(S', \boldsymbol{\tau}')$ on $n > 5$ leaves as follows. We take any tree topology $S'$ on $n-4$ leaves and replace one of the leaves by $S$. For $\boldsymbol{\tau}'$, we keep the interval lengths as induced by $\boldsymbol{\tau}$ and choose all remaining interval lengths such that $\mathcal{S'}=(S', \boldsymbol{\tau}')$ is an ultrametric species tree. We now argue that $\cS'$ has the properties that $FP_{\cS'}(x_3) > FP_{\cS'}(x_1)$, but $\bE[FP(x_3)\vert \cS'] < \bE[FP(x_1)\vert \cS']$.

First, as $\cS$ replaced a leaf of $\cS'$, all edges on the path from the root of $\cS'$ to the root of $\cS$ contribute the same amount, say $c$, to the FP index of any leaf in $\{x_1, \ldots, x_5\}$. Moreover, no edge lengths within the subtree $\cS$ of $\cS'$ were changed. Hence, we have $FP_{\cS'}(x_i)=FP_{\cS}(x_i)+c$ for all $x_i \in \{x_1, \ldots, x_5\}$ and thus in particular $FP_{\cS'}(x_3) > FP_{\cS'}(x_1)$.

In order to show the second assertion, we now consider the decomposition of the expected FP index into the expected chopped FP index and expected top FP index as given in Equation~\eqref{FP_Chop_Top_Decomp} with respect to $\cS'$ and with respect to cut height $\eta = \tau_0 + \tau_1 + \tau_2 + \tau_3$ (i.e., with respect to the same cut height as discussed for $\cS$). If we cut all histories $h \in \cH(\cS')$ at height $\eta$, the expected chopped FP indices for taxa $x_1, \ldots, x_5$ on $\cS'$ are identical to the ones observed when cutting all histories evolving within $\cS$ at height $\eta$. In particular, we have
\begin{align*}
    \bE[FP_{\text{chop}}^\eta(x_1)\vert \cS'] &\approx 2.5661 >  2.5641 \approx
    \bE[FP_{\text{chop}}^\eta(x_3)\vert \cS'].
\end{align*}
It remains to consider $\bE[FP_\text{top}^\eta(x_1)\vert \cS']$ and $\bE[FP_\text{top}^\eta(x_3)\vert \cS']$. Instead of calculating these expected values explicitly, we will consider $\bE[FP_\text{top}^\eta(x_1)\vert \cS']-\bE[FP_\text{top}^\eta(x_3)\vert \cS'] = \bE[FP_\text{top}^\eta(x_1) - FP_\text{top}^\eta(x_3) \vert \cS']$ and show that this expectation is strictly positive. Together with the fact that $\bE[FP_{\text{chop}}^\eta(x_1) \vert \cS'] > \bE[FP_{\text{chop}}^\eta(x_3) \vert \cS'] $ this will complete the proof. 

For each history $h \in \cH(\cS')$ with top part $h^{\mathit{top}}$, let $l_1$ denote the leaf of $h^\mathit{top}$ corresponding to the last ancestor of $x_1$ and let $l_3$ denote the leaf of $h^{\mathit{top}}$ corresponding to the last ancestor of $x_3$. Then, $l_1$ and $l_3$ are exchangeable in the sense that they either (a) form a cherry of $h^{\mathit{top}}$; or (b) if a top part history $\widetilde{h}^\mathit{top}$ does not contain the cherry $[l_1,l_3]$, there exists a top part history $\widehat{h}^\mathit{top}$ identical to $\widetilde{h}^\mathit{top}$ except that the roles of $l_1$ and $l_3$ are interchanged and such that for the corresponding complete histories $\widehat{h}$ and $\widetilde{h}$, we have $\boldsymbol{t}_{\widehat{h}} = \boldsymbol{t}_{\widetilde{h}}$ and $f_{\widehat{h}}(\widehat{h},\boldsymbol{t}_{\widehat{h}}) = f_{\widetilde{h}}(\widetilde{h},\boldsymbol{t}_{\widetilde{h}})$. 
Now, recall the extended top part $h^\mathit{top}_\mathit{ext}$ of a history $h$ is obtained by attaching to each leaf of $h^\mathit{top}$ the number of descending leaves it had in $h$ by edges of length 0. Moreover, recall that $\bE[FP_\text{top}^\eta(x_i)\vert \cS'] = \sum_{h \in \cH(\cS')} \bE[FP_{h^\mathit{top}_\mathit{ext}}(v_{x_i})\vert \cS']$, where $v_{x_i}$ is one of the newly attached leaves of $h^\mathit{top}_\mathit{ext}$ whose parent is an ancestor of $x_i$ in $h$. We thus need to analyze $ \sum_{h \in \cH(\cS')} \bE[FP_{h^\mathit{top}_\mathit{ext}}(v_{x_1}) - FP_{h^\mathit{top}_\mathit{ext}}(v_{x_3})\vert \cS']$. For this purpose, we partition the set $\cH(\cS')$ of histories associated with $\cS'$ into two disjoint subsets: (a) the set of histories whose induced top parts contain the cherry $[l_1,l_3]$, and (b) the set of histories whose induced top parts do not contain the cherry $[l_1,l_3]$. 
In either case, this will lead to six subcases for the extended top part histories depending on the number of leaves attached to $l_1$ and to $l_3$, respectively. 

\begin{enumerate}[(a)]
\item We first consider the subset of histories in $\cH(\cS')$ whose induced top parts contain the cherry $[l_1,l_3]$. Consider such a top part and denote the time of coalescence of $l_1$ and $l_3$ by $t_{13}$. Now, note that each such fixed top part gives rise to a distribution of extended top parts obtained by attaching $n_1 \in \{1,2\}$ leaves to $l_1$ and $n_3 \in \{1,2,3\}$ leaves to $l_3$ (see Figure \ref{Fig_CherryCase}), where $n_1$ and $n_3$ are distributed according to Table \ref{Table_ni_dist} (due to the fact that we cut all histories in $\cH(\cS')$ at height $\eta$, i.e., at the same cut height for which this distribution was obtained when cutting all histories in $\cH(\cS)$). Moreover, note that $n_1$ and $n_3$ are clearly independent. Referring to Figure \ref{Fig_CherryCase}, among the histories in $\cH(\cS')$ whose induced top parts contain the cherry $[l_1,l_3]$ and the time of coalescence of $l_1$ and $l_3$ is $t_{13}$, a proportion of $P(n_1=1) \cdot P(n_3=1)$ of extended top parts is of type (i), a proportion of $P(n_1=1) \cdot P(n_3=2)$ is of type (ii) and so forth. Now, considering these types of extended top parts more in-depth, we observe the following:
\begin{itemize}
    \item If the extended top part of a history $h$ is of type (i) or (v), we have $FP_{h^\mathit{top}_\mathit{ext}}(v_{x_1}) - FP^\mathit{top}_\mathit{ext}(v_{x_3}) = 0$.
    \item If the extended top part of a history $h$ is of type (iii), we have $FP_{h^\mathit{top}_\mathit{ext}}(v_{x_1}) - FP_{h^\mathit{top}_\mathit{ext}}(v_{x_3}) = \frac{2 t_{13}}{3} > 0$. Similarly, if the extended top part of a history $h$ is of type (vi), we have $FP_{h^\mathit{top}_\mathit{ext}}(v_{x_1}) - FP_{h^\mathit{top}_\mathit{ext}}(v_{x_3}) = \frac{t_{13}}{6} > 0$.
    \item If the extended top part of a history $h$ is of type (ii), we have $FP_{h^\mathit{top}_\mathit{ext}}(v_{x_1}) - FP_{h^\mathit{top}_\mathit{ext}}(v_{x_3}) = \frac{t_{13}}{2} > 0$, and if it is of type (iv), we have $FP_{h^\mathit{top}_\mathit{ext}}(v_{x_1}) - FP_{h^\mathit{top}_\mathit{ext}}(v_{x_3}) = \frac{-t_{13}}{2} < 0$. However, as type (ii) has probability $P(n_1=1) \cdot P(n_3=2) \approx 0.218$, whereas type (iv) has probability $P(n_1=2) \cdot P(n_3=1) \approx 0.201$, the proportion of histories with extended top parts contributing $\frac{+t_{13}}{2}$ is larger than the proportion of histories with extended top parts contributing $\frac{-t_{13}}{2}$.
\end{itemize}
In summary, we can thus conclude that 
\begin{align}
    \sum\limits_{\substack{h \in \cH(\cS'): \\ h^{top} \text{ contains cherry } [l_1,l_3]}} \bE[FP_{h^\mathit{top}_\mathit{ext}}(v_{x_1}) - FP_{h^\mathit{top}_\mathit{ext}}(v_{x_3})\vert \cS'] > 0. \label{Cherry_Histories}
\end{align}

\begin{figure}[htbp]
    \centering
    \includegraphics[scale=0.2]{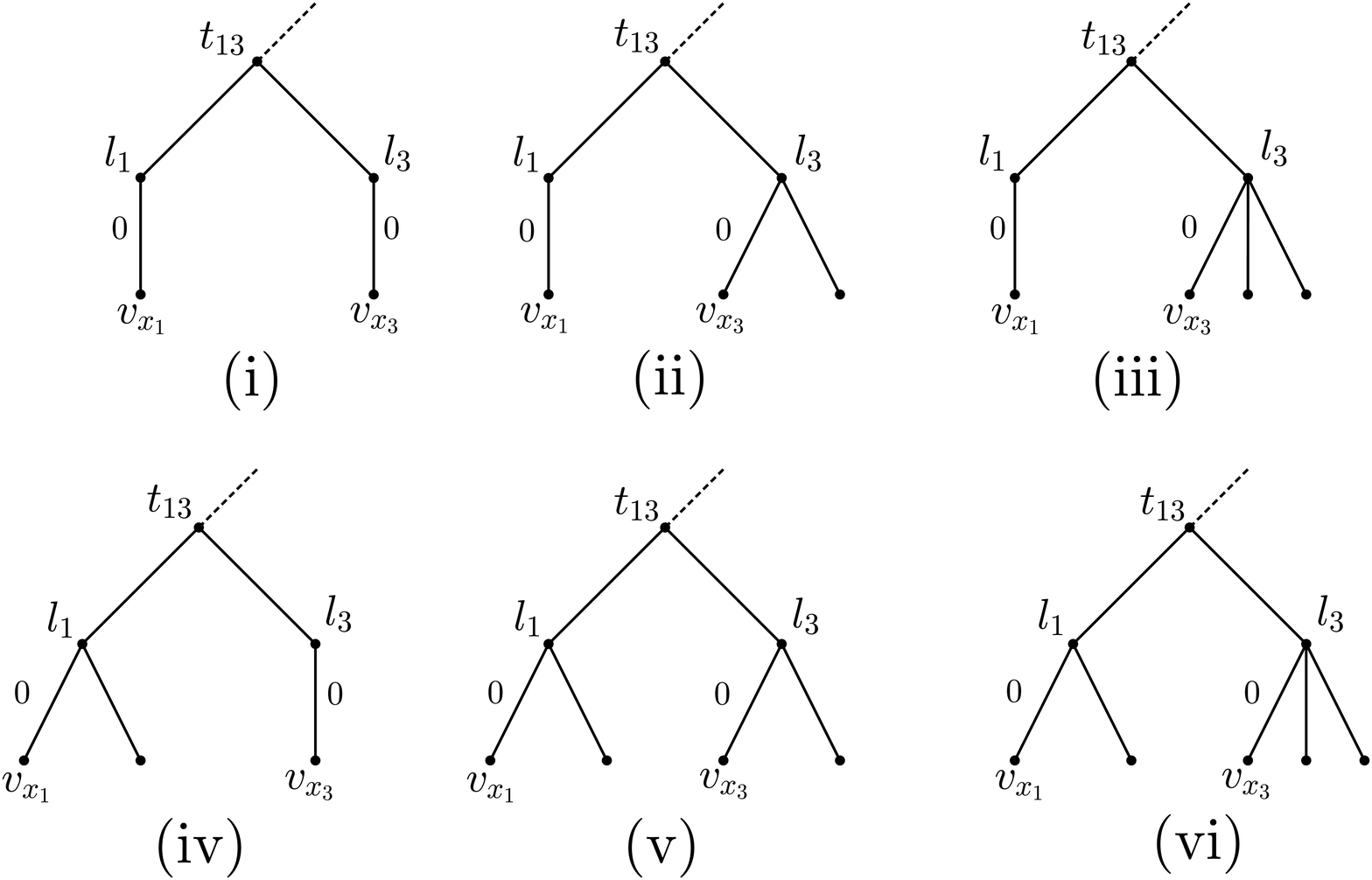}
    \caption{Types of extended top parts induced by top parts of histories containing the cherry $[l_1,l_3]$ as described in the proof of Theorem \ref{Thm_rankswap}.}
    \label{Fig_CherryCase}
\end{figure}

\item Now, consider the subset of histories in $\cH(\cS')$ whose induced top parts do not contain the cherry $[l_1,l_3]$. As explained above, for each such history $\widetilde{h}$ with top part $\widetilde{h}^{\mathit{top}}$, there exists a history $\widehat{h}$ with top part $\widehat{h}^{\mathit{top}}$ such that $\widetilde{h}^{\mathit{top}}$ and $\widehat{h}^{\mathit{top}}$ are identical up to a permutation of the leaves $l_1$ and $l_3$, and such that  $\boldsymbol{t}_{\widehat{h}} = \boldsymbol{t}_{\widetilde{h}}$ and $f_{\widehat{h}}(\widehat{h},\boldsymbol{t}_{\widehat{h}}) = f_{\widetilde{h}}(\widetilde{h},\boldsymbol{t}_{\widetilde{h}})$. Similar to case (a), each fixed pair of matching top parts gives rise to a distribution of matching pairs of extended top parts by attaching $n_1 \in \{1,2\}$ leaves to $l_1$ and $n_3 \in \{1,2,3\}$ leaves to $l_3$, where $n_1$ and $n_3$ have the distribution given in Table \ref{Table_ni_dist}. Using the same enumeration as in case (a), we thus have the following subcases:
\begin{enumerate}[(i)]
    \item $n_1 = n_3 = 1$;
    \item $n_1 = 1$ and $n_3=2$;
    \item $n_1 = 1$ and $n_3=3$;
    \item $n_1= 2$ and $n_3 = 1$;
    \item $n_1 = 2$ and $n_3 = 2$;
    \item $n_1 = 2$ and $n_3 = 3$.
\end{enumerate}
Referring to this enumeration scheme and Figure \ref{Fig_NonCherryCase} and fixing the time of coalescence leading to $lca(x_1,x_3)$, among the matching pairs of histories in $\cH(\cS)$ with induced top parts $\widetilde{h}^{\mathit{top}}$ and $\widehat{h}^{\mathit{top}}$, respectively, a proportion of $P(n_1=1) \cdot P(n_3=1)$ of pairs of extended top parts is of type (i), a proportion of $P(n_1=1) \cdot P(n_3=2)$ is of type (ii) and so forth. Similar to case (a), we now analyze these types more in-depth, and observe the following:
\begin{itemize}
    \item For pairs of matching histories $(\widetilde{h},\widehat{h})$ with extended top parts of type (i) or (v), we have
    \begin{align*}
        \underbrace{FP_{\widetilde{h}^\mathit{top}_\mathit{ext}}(v_{x_1}) - FP_{\widehat{h}^\mathit{top}_\mathit{ext}}(v_{x_3})}_{=0} + \underbrace{FP_{\widehat{h}^\mathit{top}_\mathit{ext}}(v_{x_1}) - FP_{\widetilde{h}^\mathit{top}_\mathit{ext}}(v_{x_3})}_{=0} = 0.
    \end{align*}
    \item For pairs of matching histories $(\widetilde{h},\widehat{h})$ with extended top parts of type (iii) or (vi), we have
    \begin{align*}
        \underbrace{FP_{\widetilde{h}^\mathit{top}_\mathit{ext}}(v_{x_1}) - FP_{\widehat{h}^\mathit{top}_\mathit{ext}}(v_{x_3})}_{> 0} + \underbrace{FP_{\widehat{h}^\mathit{top}_\mathit{ext}}(v_{x_1}) - FP_{\widetilde{h}^\mathit{top}_\mathit{ext}}(v_{x_3})}_{> 0} > 0.
    \end{align*}
  \item For pairs of matching histories $(\widetilde{h},\widehat{h})$ with extended top parts of type (ii), we have
  \begin{align*}
        \underbrace{FP_{\widetilde{h}^\mathit{top}_\mathit{ext}}(v_{x_1}) - FP_{\widehat{h}^\mathit{top}_\mathit{ext}}(v_{x_3})}_{> 0} + \underbrace{FP_{\widehat{h}^\mathit{top}_\mathit{ext}}(v_{x_1}) - FP_{\widetilde{h}^\mathit{top}_\mathit{ext}}(v_{x_3})}_{> 0} \eqqcolon \delta > 0,
    \end{align*} 
    whereas for pairs of matching histories $(\widetilde{h},\widehat{h})$ with extended top parts of type (iv), we have 
    \begin{align*}
        \underbrace{FP_{\widetilde{h}^\mathit{top}_\mathit{ext}}(v_{x_1}) - FP_{\widehat{h}^\mathit{top}_\mathit{ext}}(v_{x_3})}_{< 0} + \underbrace{FP_{\widehat{h}^\mathit{top}_\mathit{ext}}(v_{x_1}) - FP_{\widetilde{h}^\mathit{top}_\mathit{ext}}(v_{x_3})}_{< 0} = - \delta < 0.
    \end{align*} 
    However, as in Case (a), pairs of type (ii) have a higher probability than pairs of type (iv), and thus the proportion of matching pairs of histories with extended top parts contributing a positive value is larger than the proportion of histories with extended top parts contributing the same negative value. 
\end{itemize}
As we can partition the subset of histories in $\cH(\cS')$ whose induced top parts do not contain the cherry $[l_1,l_3]$ in pairs of matching histories $(\widetilde{h}, \widehat{h})$ as described above, we can thus also conclude that
\begin{align}
      \sum\limits_{\substack{h \in \cH(\cS'): \\ h^{top} \text{ does not contain cherry } [l_1,l_3]}} \bE[FP_{h^\mathit{top}_\mathit{ext}}(v_{x_1}) - FP_{h^\mathit{top}_\mathit{ext}}(v_{x_3})\vert \cS'] > 0. \label{NonCherry_Histories}
\end{align}
\end{enumerate}

\begin{figure}[htbp]
    \centering
    \includegraphics[scale=0.15]{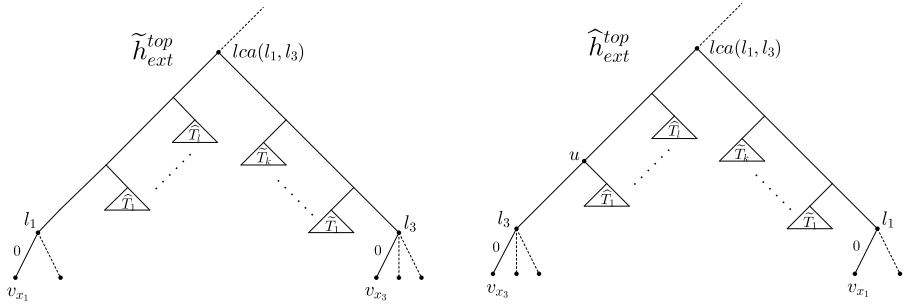}
    \caption{Types of pairs of extended top parts induced by pairs of top parts of histories not containing the cherry $[l_1,l_3]$ as described in the proof of Theorem \ref{Thm_rankswap}. The dashed lines indicate that $v_{x_1}$ may have one sibling and $v_{x_3}$ may have one or two siblings.}
    \label{Fig_NonCherryCase}
\end{figure}

Taking Equations \eqref{Cherry_Histories} and \eqref{NonCherry_Histories} together, we obtain
\begin{align*}
    \bE[FP_{\text{top}}^\eta (x_1) - FP_{\text{top}}^\eta (x_3) \vert \cS']
    &= \sum\limits_{h \in \cH(\cS')} \bE[FP_{h^\mathit{top}_\mathit{ext}}(v_{x_1}) - FP_{h^\mathit{top}_\mathit{ext}}(v_{x_3})\vert \cS']  > 0
\end{align*}
as claimed. This completes the proof.
\end{proof}

\paragraph{Remarks.} We conclude this section with a few remarks. First, note that the proof of Theorem~\ref{Thm_rankswap} establishes a slightly stronger result than the statement of Theorem~\ref{Thm_rankswap}. More precisely, we not only prove the existence of \emph{one} species tree $\cS$ on $n \geq 5$ leaves with the desired property, but in fact establish this result for a \emph{family} of such trees for $n$ large enough, as we did not fix the tree on $n-4$ leaves used in the construction, such that any tree on $n-4$ leaves can be used.\\

Second, note that in the proof of Theorem~\ref{Thm_rankswap}, we used a very particular choice of $\boldsymbol{\tau}$ leading to $FP_\cS(x_3) > FP_\cS(x_1)$ and $\bE[FP(x_3)\vert \cS] < \bE[FP(x_1)\vert \cS]$ for the species tree $\cS$ on $n=5$ leaves depicted in Figure \ref{Fig_5Leaves_Notation}. However, there are infinitely many numerical values for $\boldsymbol{\tau}$ with this property and we depict slices of the corresponding regions of parameter (branch length) space in Figure \ref{Fig_ParameterPlots}. Note, however, that these regions were solely plotted by considering $FP_\cS(x_3) > FP_\cS(x_1)$ and $\bE[FP(x_3)\vert \cS] < \bE[FP(x_1)\vert \cS]$ for the species tree $\cS$ on $n=5$ leaves depicted in Figure \ref{Fig_5Leaves_Notation}.\footnote{In particular, we did not re-compute the probability distribution corresponding to Table \ref{Table_ni_dist} and it might be the case that not every point in the plotted region has the property that $P(n_1=1) \cdot P(n_3=2) > P(n_1=2) \cdot P(n_3=1)$, which was used in the proof of Theorem \ref{Thm_rankswap} for $n > 5$.}

\begin{figure}[htbp]
    \centering
    \begin{subfigure}{0.3\textwidth}
    \includegraphics[scale=0.35]{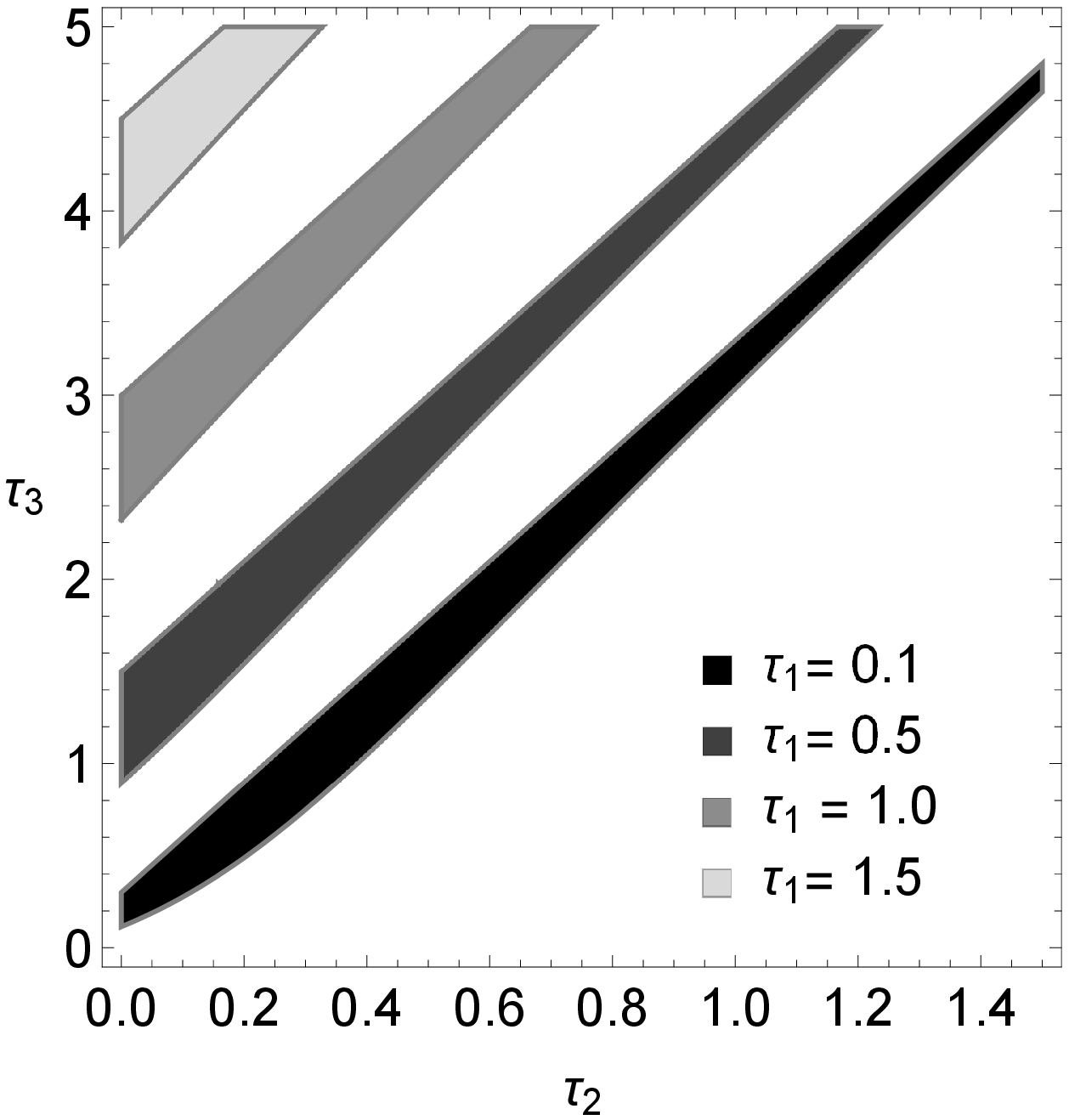}
    \caption{$\tau_1$ fixed.}
    \end{subfigure}
    \begin{subfigure}{0.3\textwidth}
    \includegraphics[scale=0.35]{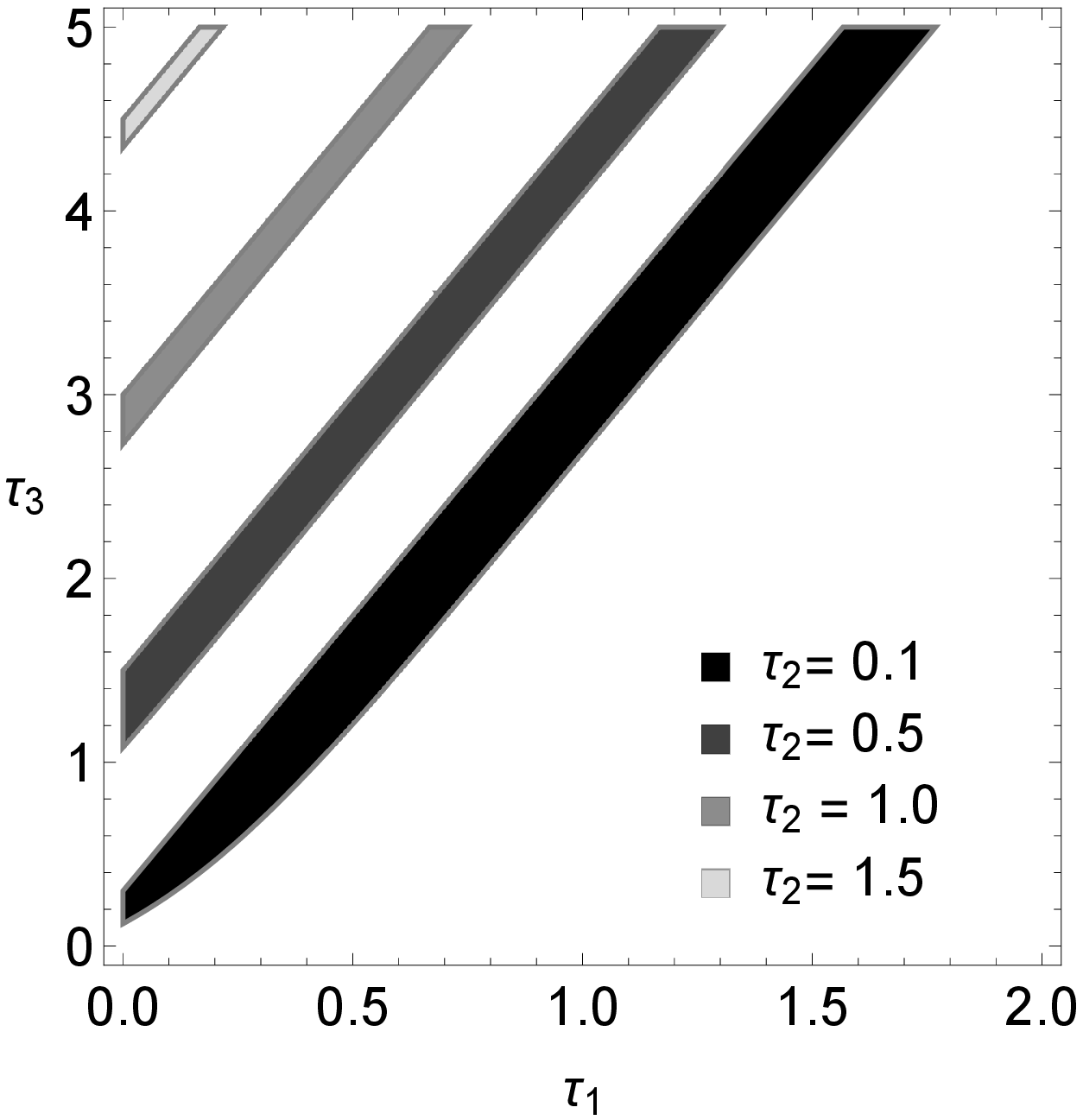}
    \caption{$\tau_2$ fixed.}
    \end{subfigure}
    \begin{subfigure}{0.3\textwidth}
    \includegraphics[scale=0.35]{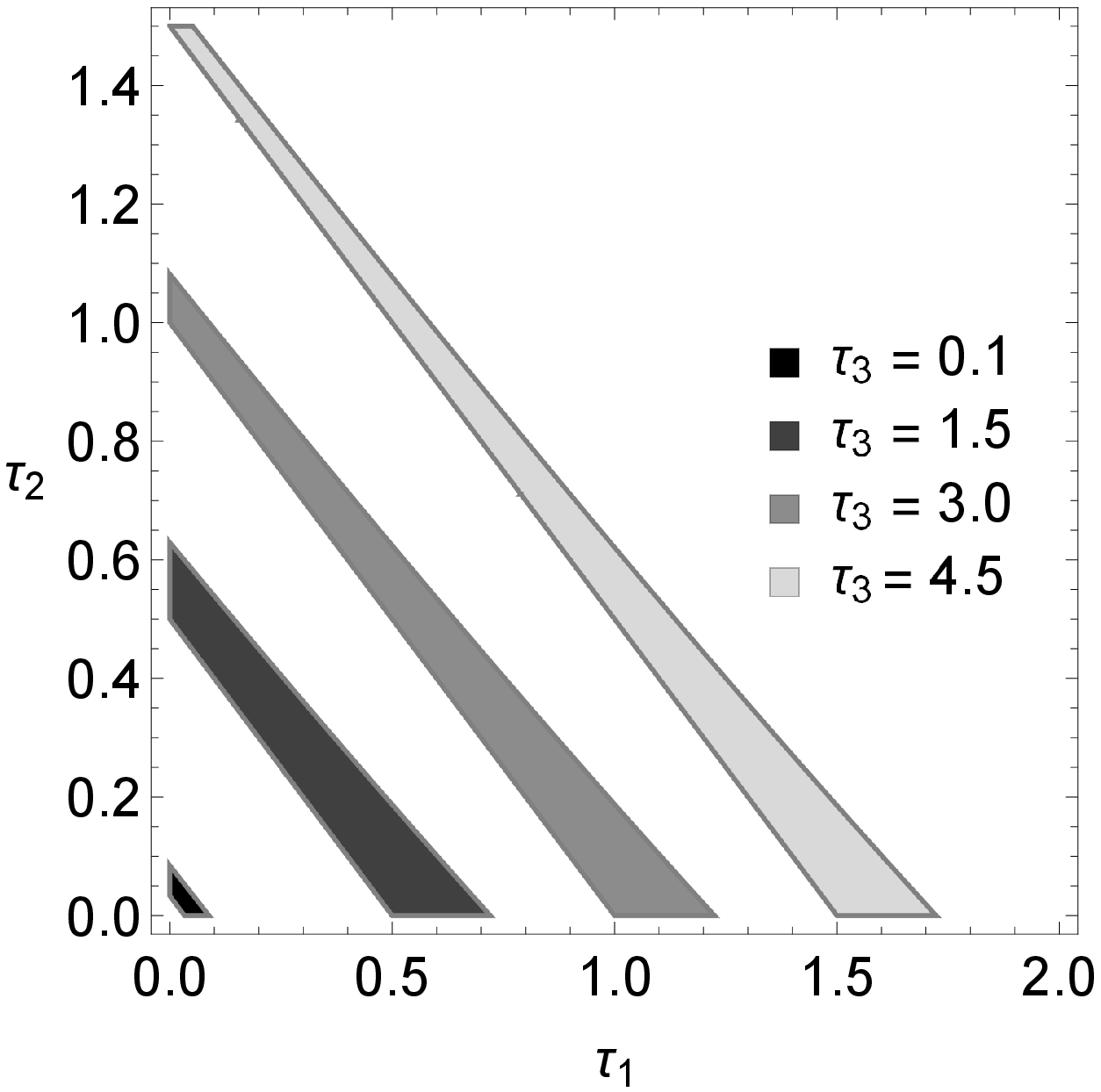}
    \caption{$\tau_3$ fixed.}
    \end{subfigure}
    \caption{Slices of parameter regions (shaded) for the species tree $\cS$ depicted in Figure \ref{Fig_5Leaves_Notation} such that $FP_{\cS}(x_3) > FP_{\cS}(x_1)$, whereas $\bE[FP(x_3)\vert \cS] < \bE[FP(x_1)\vert \cS]$. In each subfigure, a different parameter is fixed and four different choices are shown, respectively.}
    \label{Fig_ParameterPlots}
\end{figure}

Third, we note that there exist two additional species trees on 5 leaves, $\cS'$ and $\cS''$ say, with the same topology as the species tree $\cS$ used in our construction, but with a different order of speciation events (see Figure \ref{Fig_5LeavesAlternativeRankings}). Expressions for the FP indices on the species trees $\cS'$ and $\cS''$ and the respective expected FP indices are given in the Appendix. It is an interesting question whether a rank swap can also occur for these two species trees. In case of the species tree $\cS'$, we can answer this affirmatively. For instance, for $\boldsymbol{\tau}=(1.0000, 0.0050, 0.0075, 0.0100)$, we have $FP_{\cS'}(x_3) \approx 1.01583 > 1.01375 = FP_{\cS'}(x_1)$, but $\bE[FP(x_3)\vert \cS'] \approx 1.8510951 < 1.8518378 \approx \bE[FP(x_1)\vert \cS']$. For the species tree $\cS''$, on the other hand, we conjecture that a rank swap is not possible, but leave a proof or counterexample to this conjecture to future research.\\

\begin{figure}[htbp]
    \centering
    \includegraphics[scale=0.175]{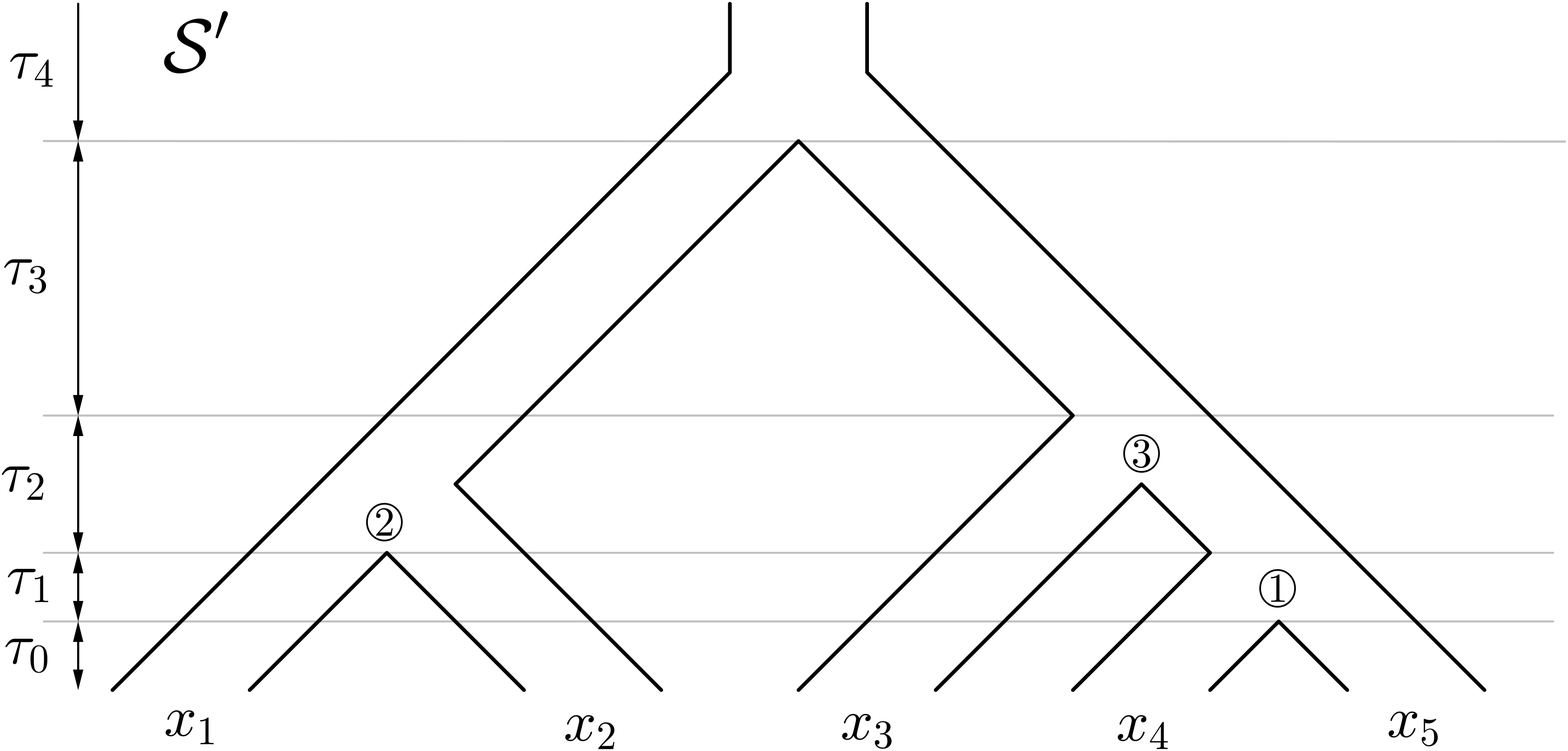} \\
    \vspace{0.1cm}
    \includegraphics[scale=0.175]{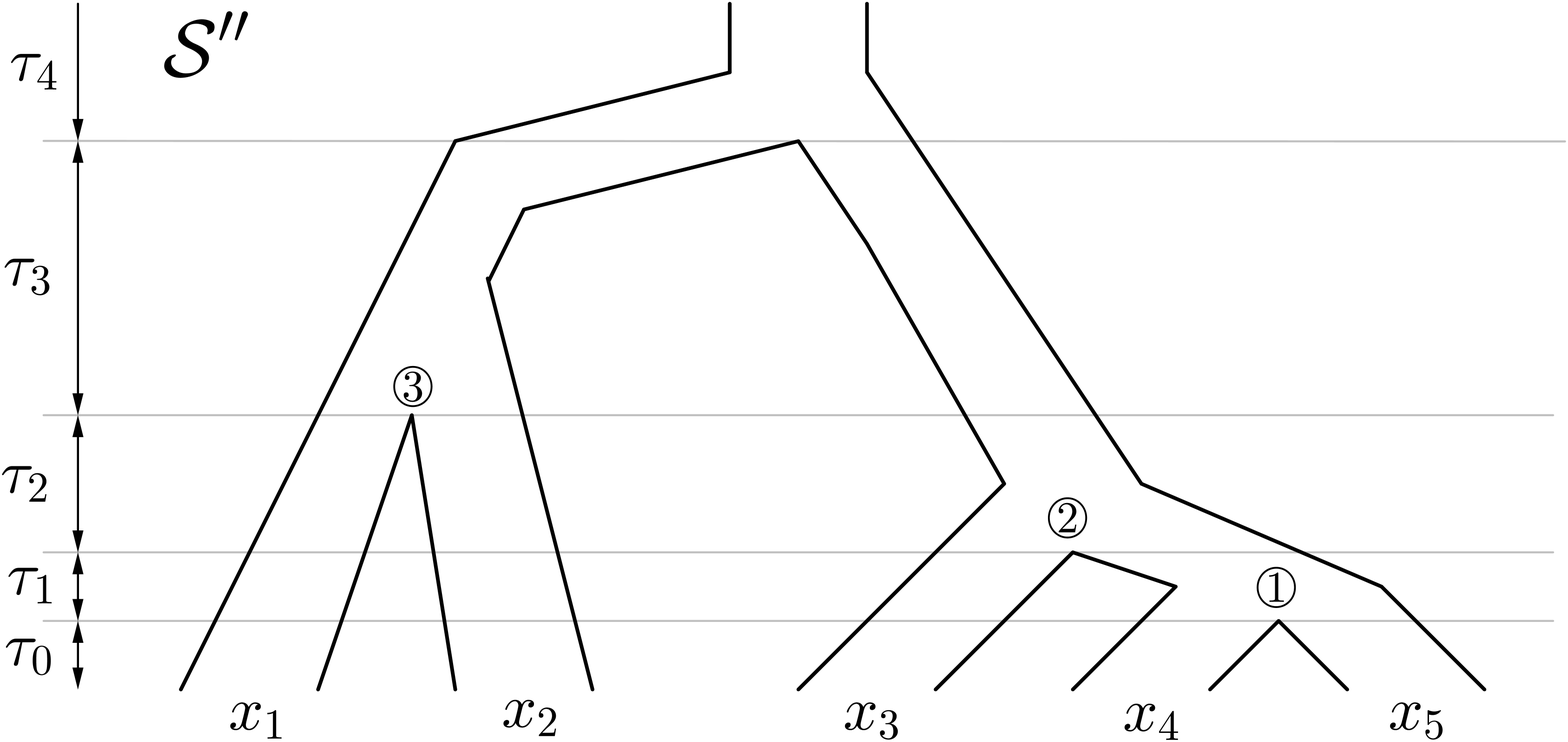} \\
    \caption{Species trees $\cS'$ and $\cS''$ that have the same unranked topology as the species tree $\cS$ depicted in Figure \ref{Fig_5Leaves_Notation}, but a different ranked topology. In particular, they depict a different order of speciation events (compare the circled numbers to see this).} 
    \label{Fig_5LeavesAlternativeRankings}
\end{figure}

Finally, we remark that there is a curious connection between some of the results presented in this manuscript and results obtained by \citet{Degnan2012a,Degnan2012} in relation to the existence of so-called \emph{anomalous ranked gene trees} (ARGTs), whereby a ranked gene tree history that does not match the ranked species tree can have a greater probability than the matching ranked gene tree. More precisely, \citet{Degnan2012a} showed that ARGTs do not exist for 3-taxon or 4-taxon species trees, and \citet{Degnan2012} showed that more generally ARGTs do not exist for caterpillar and pseudocaterpillar species trees. This directly parallels Corollary \ref{Cor_atmost4} and Theorem \ref{Thm_Caterpillar_Pseudocaterpillar} in the present manuscript, which state that the ranking induced by the FP index on the species tree and the ranking induced by the expected FP index cannot differ in these cases. Moreover, the species tree $\cS$ depicted in Figure \ref{Fig_5Leaves_Notation}, i.e., the species tree we used in the proof of Theorem \ref{Thm_rankswap} to show that there are species trees for which the rankings differ, is exactly the ranked species tree used by \cite{Degnan2012a, Degnan2012} to prove the existence of ARGTs for 5-taxon trees. As indicated above, rank swaps for the FP index exist for at least one other ranking of this species tree topology (namely for the ranked species tree $\cS'$ shown in Figure \ref{Fig_5LeavesAlternativeRankings}) and it would be interesting to see if this again parallels or maybe contrasts the case of ARGTs. Based on the characterization of species trees producing ARGTs given by \cite{Degnan2012} it is unfortunately not possible to decide if the two species trees $\cS'$ and $\cS''$ induce ARGTs.

\section{Empirical data}\label{Sec_Empirical}
In order to illustrate the variability of prioritization orders obtained from the FP index for empirical data, we considered the multilocus mammal data set of \citet{Song2012}. As the original collection of gene trees for the data was criticized in the literature for methodological errors, we used those from a re-analysis of \citet{Springer2016}. These data consist of 447 gene trees for 37 species (33 of which are mammals and four of which are outgroup species; a list of species names and their abbreviations can be found in Table \ref{tab_names} in the appendix).\footnote{Note that we use this data set as a mere example for the variability of the FP rankings across gene trees and do not attempt to make any statement about the conservation implications for any of those species.} As the gene trees provided by \citet{Springer2016} were unrooted, we rooted them using \textit{Gallus gallus} as the outgroup. We then calculated the FP index for all 37 species on all 447 gene trees. In addition, we calculated the average FP index of each species across the 447 gene trees. We also calculated the FP index for all species using two methods: the species tree estimation method SVDQuartets \citep{Chifman2014} as implemented in the PAUP* package \citep{swofford-paup} and the maximum likelihood method for the concatenated data using RAxML version 8.2.10 \citep{stamatakis-raxml}. Maximum likelihood branch lengths were computed for both trees under the GTR+$\Gamma$ model in PAUP*, additionally enforcing a molecular clock. For reproducibility, we provide the exact commands used in the appendix. Finally, the FP indices on the 447 gene trees and the two species trees as well as the average FP indices were transformed into rankings and a boxplot summarizing the results was generated (Figure \ref{Fig_CombinedBoxplot}).

As evident from the boxplot, there is a huge variability of ranking positions across the 447 gene trees for most species. In particular, there are species such as \textit{Sus scrofa} (Sus), whose ranking positions span the entire range of ranking positions possible or are close to it (as in the case of \textit{Callithrix jacchus} (Cal), \textit{Equus caballus} (Equ), or \textit{Macaca mulatta} (New)). On the other hand, there are some species such as \textit{Gorilla gorilla} (Gor), \textit{Homo sapiens} (Hom), \textit{Pan troglodytes} (Pan), and \textit{Pongo pygmaeus} (Pon) that show comparatively little variability in their ranking positions and consistently receive low ranks. These four species also show the interesting property that the ranks obtained from the average FP index across the 447 gene trees and the ranks obtained from the two species tree estimates all coincide except for one case (for Pon, the rank from the mean over gene trees differs slightly from the rank on the two species tree estimates). In general, it is interesting to note that the average FP index and the FP index on the two species tree estimates yield similar ranks in many cases. In particular, the rank of the SVDQuartets tree almost always coincides with either the rank based on the average FP index or the rank based on the concatenated data. However, there are nevertheless some interesting rank swaps. For instance, considering species \textit{Procavia capensis} (Pro) and \textit{Tupaia belangeri} (Tup), based on the average FP index,  \textit{Tupaia belangeri} receives are higher rank than \textit{Procavia capensis}, but on both species tree estimates, \textit{Procavia capensis} receives a higher rank than  \textit{Tupaia belangeri}. This last observation highlights that our theoretical results on possible rank swaps between the FP index on a species tree and the expected FP index across all gene trees associated with it (Theorem \ref{Thm_rankswap}) might occur in practice and need to be considered in the process of conservation planning.

\begin{sidewaysfigure}
    \centering
    \includegraphics[scale=0.5]{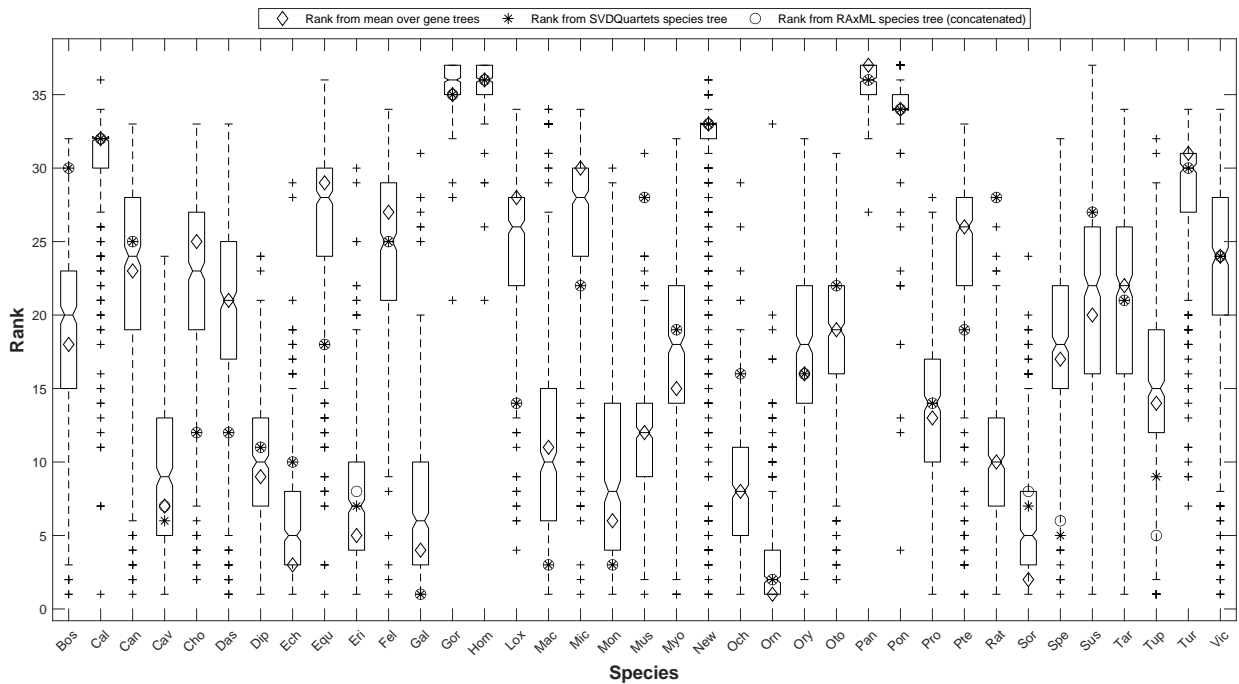}
    \caption{Boxplot of the ranking positions obtained from the FP index for the multilocus mammal data set of \citet{Song2012} consisting of 447 gene trees for 37 species (a table listing the full scientific and common names of all species can be found in the appendix). In addition, the ranking position obtained from the average FP index across the 447 gene trees (diamonds), the ranking position obtained from the FP index for a species tree estimated with SVDQuartets \citep{Chifman2014} (stars), and a species tree estimated with RAxML version 8.2.10 \citep{stamatakis-raxml} for the concatenated data (circles) are depicted.}
    \label{Fig_CombinedBoxplot}
\end{sidewaysfigure}

\section{Concluding remarks} \label{Sec_Discussion}
The FP index is a popular phylogenetic diversity index used to prioritize species for conservation. It assesses the importance of species for overall biodiversity based on their placement in an underlying phylogenetic tree. As such, the prioritization order obtained from the FP index heavily depends on the topology and branch lengths of the phylogenetic tree used. In particular, different phylogenetic trees for the same set of species might lead to different prioritization orders. 

In this note, we analyzed the variability of prioritization orders obtained from the FP index for different phylogenetic trees on the same set of species in the context of gene tree species tree discordance. More precisely, we compared the ranking obtained from the FP index on a given species tree with the ranking obtained from the expected FP index across all gene tree histories evolving within the species tree under the multispecies coalescent model. On one hand, we showed that the two prioritization orders are identical for all trees with at most 4 leaves, and more generally for all caterpillar and pseudocaterpillar trees. On the other hand, we showed that for all $n \geq 5$, there exists a species tree on $n$ leaves such that the two rankings are different. We ended by illustrating the variability in ranking orders obtained from the FP index for an empirical multilocus mammal data set. Here, we observed a high variability in ranking positions for most species, indicating that if conservation decisions are based on single genes, the prioritization order may be very different depending on which gene is used. We also observed rank swaps between the average FP index across gene trees and the FP index on two species tree estimates, empirically underpinning our theoretical results stated in Theorem \ref{Thm_rankswap}.

Our work suggests several broad questions that might be interesting to explore in the future. From a theoretical point of view, it would be interesting to analyze if there are classes of tree topologies other than the caterpillar and pseudocaterpillar trees for which the FP index on the species tree and the expected FP index across all gene tree histories necessarily induce the same prioritization order. Moreover, in cases for which there are differences between the two prioritization orders, it would be interesting to assess how different the two rankings can be in the most extreme cases. In particular, we remark that our construction for the existence of a species tree on $n \geq 5$ leaves leading to two different ranking orders relies on a switch in ranking positions among the leaves of a certain 5-taxon subtree, and if $n$ is large and this is the only difference between the two rankings, this might be considered a minor difference. It would thus be interesting to see if there are other constructions that lead to more severe differences between the two rankings. From an empirical point of view, it would also be interesting to study the variability of the prioritization orders obtained from the FP index for different gene trees and species tree estimates for a broader class of data sets, in particular for data sets containing conservation-relevant species, i.e., species at high risk of extinction. More generally, as both our theoretical results as well the analysis of the empirical data set show that different prioritization orders may be obtained from considering the FP index on a given species tree, on a particular gene tree, or even on all possible gene trees associated with the species tree, an immediate question that arises is the question of which evolutionary information should be used as the basis for conservation decisions. Should these decisions be based on the species tree, on a collection of gene trees, or taking into account a combination of both? We leave all of these questions to future research, but remark that particularly answering the last question will likely require joint efforts from theoreticians and empiricists.

\section*{Acknowledgements}
MF was supported by the joint research project \textit{\textbf{DIG-IT!}}
funded by the European Social Fund (ESF), reference: ESF/14-BM-A55-
0017/19, and the Ministry of Education, Science and Culture of Mecklenburg-Vorpommern, Germany. 
KW was supported by The Ohio State University's President's Postdoctoral Scholars Program. 
All authors thank three anonymous reviewers for detailed comments on an earlier version of this manuscript. Moreover, all authors thank James Degnan for a helpful discussion on the existence of anomalous ranked gene trees.

\bibliographystyle{abbrvnat}
\bibliography{References}

\begin{thebibliography}{23}
\providecommand{\natexlab}[1]{#1}
\providecommand{\url}[1]{\texttt{#1}}
\expandafter\ifx\csname urlstyle\endcsname\relax
  \providecommand{\doi}[1]{doi: #1}\else
  \providecommand{\doi}{doi: \begingroup \urlstyle{rm}\Url}\fi

\bibitem[Chifman and Kubatko(2014)]{Chifman2014}
J.~Chifman and L.~Kubatko.
\newblock Quartet inference from {SNP} data under the coalescent model.
\newblock \emph{Bioinformatics}, 30\penalty0 (23):\penalty0 3317--3324, Aug.
  2014.
\newblock \doi{10.1093/bioinformatics/btu530}.

\bibitem[Degnan and Rosenberg(2009)]{Degnan2009}
J.~H. Degnan and N.~A. Rosenberg.
\newblock Gene tree discordance, phylogenetic inference and the multispecies
  coalescent.
\newblock \emph{Trends in Ecology \& Evolution}, 24\penalty0 (6):\penalty0
  332--340, June 2009.
\newblock \doi{10.1016/j.tree.2009.01.009}.

\bibitem[Degnan et~al.(2012{\natexlab{a}})Degnan, Rosenberg, and
  Stadler]{Degnan2012}
J.~H. Degnan, N.~A. Rosenberg, and T.~Stadler.
\newblock A characterization of the set of species trees that produce anomalous
  ranked gene trees.
\newblock \emph{IEEE/ACM Transactions on Computational Biology and
  Bioinformatics}, 9\penalty0 (6):\penalty0 1558--1568, Nov.
  2012{\natexlab{a}}.
\newblock \doi{10.1109/tcbb.2012.110}.

\bibitem[Degnan et~al.(2012{\natexlab{b}})Degnan, Rosenberg, and
  Stadler]{Degnan2012a}
J.~H. Degnan, N.~A. Rosenberg, and T.~Stadler.
\newblock The probability distribution of ranked gene trees on a species tree.
\newblock \emph{Mathematical Biosciences}, 235\penalty0 (1):\penalty0 45--55,
  Jan. 2012{\natexlab{b}}.
\newblock \doi{10.1016/j.mbs.2011.10.006}.

\bibitem[Faith(1992)]{Faith1992}
D.~P. Faith.
\newblock Conservation evaluation and phylogenetic diversity.
\newblock \emph{Biological Conservation}, 61\penalty0 (1):\penalty0 1--10,
  1992.
\newblock \doi{10.1016/0006-3207(92)91201-3}.

\bibitem[Fuchs and Jin(2015)]{Fuchs2015}
M.~Fuchs and E.~Y. Jin.
\newblock Equality of {S}hapley value and {F}air {P}roportion index in
  phylogenetic trees.
\newblock \emph{Journal of Mathematical Biology}, 71\penalty0 (5):\penalty0
  1133--1147, Nov 2015.
\newblock \doi{10.1007/s00285-014-0853-0}.

\bibitem[Haake et~al.(2008)Haake, Kashiwada, and Su]{Haake2008}
C.-J. Haake, A.~Kashiwada, and F.~E. Su.
\newblock The {S}hapley value of phylogenetic trees.
\newblock \emph{Journal of Mathematical Biology}, 56\penalty0 (4):\penalty0
  479--497, 2008.
\newblock \doi{10.1007/s00285-007-0126-2}.

\bibitem[Isaac et~al.(2007)Isaac, Turvey, Collen, Waterman, and
  Baillie]{Isaac2007}
N.~J. Isaac, S.~T. Turvey, B.~Collen, C.~Waterman, and J.~E. Baillie.
\newblock Mammals on the {EDGE}: {C}onservation priorities based on threat and
  phylogeny.
\newblock \emph{PLoS ONE}, 2\penalty0 (3):\penalty0 e296, Mar 2007.

\bibitem[Maddison(1997)]{Maddison1997}
W.~P. Maddison.
\newblock Gene trees in species trees.
\newblock \emph{Systematic Biology}, 46\penalty0 (3):\penalty0 523--536, Sept.
  1997.
\newblock \doi{10.1093/sysbio/46.3.523}.

\bibitem[Nichols(2001)]{Nichols2001}
R.~Nichols.
\newblock Gene trees and species trees are not the same.
\newblock \emph{Trends in Ecology \& Evolution}, 16\penalty0 (7):\penalty0
  358--364, 2001.
\newblock \doi{https://doi.org/10.1016/S0169-5347(01)02203-0}.

\bibitem[Pamilo and Nei(1988)]{Pamilo1988}
P.~Pamilo and M.~Nei.
\newblock Relationships between gene trees and species trees.
\newblock \emph{Molecular Biology and Evolution}, 5\penalty0 (5):\penalty0
  568--583, Sep 1988.
\newblock \doi{10.1093/oxfordjournals.molbev.a040517}.

\bibitem[Redding(2003)]{Redding2003}
D.~W. Redding.
\newblock Incorporating genetic distinctness and reserve occupancy into a
  conservation priorisation approach.
\newblock Master's thesis, University Of East Anglia, Norwich, UK, 2003.

\bibitem[Redding and Mooers(2006)]{Redding2006}
D.~W. Redding and A.~{\O}. Mooers.
\newblock Incorporating evolutionary measures into conservation prioritization.
\newblock \emph{Conservation Biology}, 20\penalty0 (6):\penalty0 1670--1678,
  Dec 2006.
\newblock \doi{10.1111/j.1523-1739.2006.00555.x}.

\bibitem[Redding et~al.(2008)Redding, Hartmann, Mimoto, Bokal, DeVos, and
  Mooers]{Redding2008}
D.~W. Redding, K.~Hartmann, A.~Mimoto, D.~Bokal, M.~DeVos, and A.~Mooers.
\newblock Evolutionarily distinctive species often capture more phylogenetic
  diversity than expected.
\newblock \emph{Journal of Theoretical Biology}, 251\penalty0 (4):\penalty0
  606--615, Apr 2008.
\newblock \doi{10.1016/j.jtbi.2007.12.006}.

\bibitem[Redding et~al.(2014)Redding, Mazel, and Mooers]{Redding2014}
D.~W. Redding, F.~Mazel, and A.~{\O}. Mooers.
\newblock Measuring evolutionary isolation for conservation.
\newblock \emph{{PLoS} {ONE}}, 9\penalty0 (12):\penalty0 e113490, Dec 2014.
\newblock \doi{10.1371/journal.pone.0113490}.

\bibitem[Shapley(1953)]{Shapley1953}
L.~S. Shapley.
\newblock A value for $n$--person games.
\newblock In \emph{Contributions to the Theory of Games ({AM}-28), Volume
  {II}}, pages 307--317. Princeton University Press,, Princeton, 1953.
\newblock \doi{10.1515/9781400881970-018}.

\bibitem[Song et~al.(2012)Song, Liu, Edwards, and Wu]{Song2012}
S.~Song, L.~Liu, S.~V. Edwards, and S.~Wu.
\newblock Resolving conflict in eutherian mammal phylogeny using phylogenomics
  and the multispecies coalescent model.
\newblock \emph{Proceedings of the National Academy of Sciences}, 109\penalty0
  (37):\penalty0 14942--14947, Aug. 2012.
\newblock \doi{10.1073/pnas.1211733109}.

\bibitem[Springer and Gatesy(2016)]{Springer2016}
M.~S. Springer and J.~Gatesy.
\newblock The gene tree delusion.
\newblock \emph{Molecular Phylogenetics and Evolution}, 94:\penalty0 1--33,
  Jan. 2016.
\newblock \doi{10.1016/j.ympev.2015.07.018}.

\bibitem[Stamatakis(2014)]{stamatakis-raxml}
A.~Stamatakis.
\newblock {RAxML Version 8}: {A} tool for phylogenetic analysis and
  post-analysis of large phylogenies.
\newblock \emph{Bioinformatics}, 30:\penalty0 1312--1313, 2014.

\bibitem[Swofford(2021)]{swofford-paup}
D.~Swofford.
\newblock Paup*: {P}hylogenetic analysis using parsimony (*and other methods).
\newblock \emph{www.paup.phylosolutions.com, version 4a168}, 2021.

\bibitem[Vane-Wright et~al.(1991)Vane-Wright, Humphries, and
  Williams]{Vanewright1991}
R.~Vane-Wright, C.~Humphries, and P.~Williams.
\newblock What to protect?{\textemdash}{S}ystematics and the agony of choice.
\newblock \emph{Biological Conservation}, 55\penalty0 (3):\penalty0 235--254,
  1991.
\newblock \doi{10.1016/0006-3207(91)90030-d}.

\bibitem[Vellend et~al.(2011)Vellend, K.Cornwell, Magnuson-Ford, and
  O.Mooers]{Vellend2011}
M.~Vellend, W.~K.Cornwell, K.~Magnuson-Ford, and A.~O.Mooers.
\newblock Measuring phylogenetic biodiversity.
\newblock In A.~E. Magurran and B.~J. McGill, editors, \emph{Biological
  Diversity: Frontiers in Measurement and Assessment}, chapter~14, pages
  194--207. Oxford University Press, Oxford, 2011.
\newblock ISBN 0199580677.

\bibitem[{Wolfram Research}(2020)]{Mathematica}
I.~{Wolfram Research}.
\newblock Mathematica, {V}ersion 12.2, 2020.
\newblock URL \url{https://www.wolfram.com/mathematica}.
\newblock Champaign, IL, 2020.

\end{thebibliography}

\appendix
\section{FP indices for the species trees depicted in Figure \ref{Fig_5Leaves_Notation} and Figure \ref{Fig_5LeavesAlternativeRankings}}

In the following, we give expression for the FP indices on the species tree and the expected FP indices for the species trees $\cS$, $\cS'$, and $\cS''$ on 5 leaves considered in the main part of this manuscript (Figures \ref{Fig_5Leaves_Notation} and \ref{Fig_5LeavesAlternativeRankings}). The expected FP indices were obtained using Mathematica \citep{Mathematica} to enumerate all gene tree histories in $\cH(\cS)$, respectively $\cH(\cS')$ and $\cH(\cS'')$.

\paragraph{Species tree \texorpdfstring{$\cS$}{S} (Figure \ref{Fig_5Leaves_Notation}).} 
\begin{itemize}
    \item FP indices on the species tree:
    \begin{align*}
    FP_{\cS}(x_1) &= \tau_0 + \frac{\tau_1}{2} + \frac{\tau_2}{2} + \frac{\tau_3}{2} = FP_{\cS}(x_2) \\
    FP_{\cS}(x_3) &= \tau_0 + \tau_1 + \tau_2 + \frac{\tau_3}{3} \\
    FP_{\cS}(x_4) &= \tau_0 + \tau_1 + \frac{\tau_2}{2} + \frac{\tau_3}{3} =  FP_{\cS}(x_5).
\end{align*}

\item Expected FP indices: 
\begin{align*}
    \bE\left[FP(x_1) \vert \cS\right] &= \tau_0 + \frac{\tau_1}{2} + \frac{\tau_2}{2} + \frac{\tau_3}{2} + 1 
      - \frac{ e^{-\tau_1 - 2 \tau_2 - 2 \tau_3}}{24} - \frac{e^{- \tau_1 - \tau_2 - 2 \tau_3}}{12} \\
     &\quad + \frac{e^{-\tau_2 - \tau_3}}{72} - \frac{e^{-\tau_1 - \tau_2 - \tau_3}}{12} + \frac{e^{-\tau_3}}{36} = \bE\left[FP(x_2) \vert \cS\right] \\
    \bE \left[FP(x_3) \vert \cS\right] &= \tau_0 + \tau_1 + \tau_2 + \frac{\tau_3}{3} + 1 - \frac{e^{-\tau_2}}{9} - \frac{e^{-\tau_1 - \tau_2 - 2 \tau_3}}{12} \\ 
    &\quad + \frac{e^{-\tau_2 - \tau_3}}{12} + \frac{e^{-\tau_1 - \tau_2 - \tau_3}}{18} - \frac{e^{-\tau_3}}{9}
    \\
    \bE \left[FP(x_4) \vert \cS\right] &= \tau_0 + \tau_1 + \frac{\tau_2}{2} + \frac{\tau_3}{3} + 1
    - \frac{e^{-\tau_2}}{9} - \frac{e^{-\tau_1 - 2 \tau_2 - 2 \tau_3}}{24} \\
    &\quad - \frac{e^{-\tau_1 - \tau_2 - 2 \tau_3}}{24} - \frac{e^{-\tau_2-\tau_3}}{18} + \frac{e^{-\tau_1-\tau_2-\tau_3}}{18} + \frac{e^{-\tau_3}}{36} 
    =  \bE \left[FP(x_5) \vert \cS\right].  
\end{align*}
\end{itemize}

\paragraph{Species tree \texorpdfstring{$\cS'$}{S'} (Figure \ref{Fig_5LeavesAlternativeRankings}, top)}
\begin{itemize}
    \item FP indices on the species tree:
      \begin{align*}
        FP_{\cS'}(x_1) &= \tau_0 + \tau_1 + \frac{\tau_2}{2} + \frac{\tau_3}{2} = FP_{\cS'}(x_2) \\
        FP_{\cS'}(x_3) &= \tau_0 + \tau_1 + \tau_2 + \frac{\tau_3}{3} \\
        FP_{\cS'}(x_4) &= \tau_0 + \frac{\tau_1}{2} + \frac{\tau_2}{2} + \frac{\tau_3}{3} =  FP_{\cS'}(x_5).
        \end{align*}
    \item Expected FP indices:
    \begin{align*}
         \bE\left[FP(x_1) \vert \cS'\right] &=  \tau_0 + \tau_1 + \frac{\tau_2}{2} + \frac{\tau_3}{2} + 1 - \frac{e^{-\tau_1-2\tau_2-2 \tau_3}}{24} - \frac{e^{-\tau_2 - 2 \tau_3}}{12} \\
         &\quad - \frac{e^{-\tau_2-\tau_3}}{12} + \frac{e^{-\tau_1 - \tau_2 - \tau_3}}{72} + \frac{e^{-\tau_3}}{36} = \bE\left[FP(x_2) \vert \cS'\right] \\
           \bE\left[FP(x_3) \vert \cS'\right] &= \tau_0 + \tau_1 + \tau_2 + \frac{\tau_3}{3} + 1 - \frac{e^{-\tau_1-\tau_2}}{9} - \frac{e^{-\tau_2-2 \tau_3}}{12} \\
           &\quad + \frac{e^{-\tau_2-\tau_3}}{18} + \frac{e^{-\tau_1-\tau_2-\tau_3}}{12} - \frac{e^{-\tau_3}}{9} \\
         \bE\left[FP(x_4) \vert \cS'\right] &= \tau_0 + \frac{\tau_1}{2} + \frac{\tau_2}{2} + \frac{\tau_3}{3} + 1 - \frac{e^{-\tau_1-\tau_2}}{9} - \frac{e^{-\tau_1-2\tau_2-2\tau_3}}{24} \\
         &\quad - \frac{e^{-\tau_2-2 \tau_3}}{24} + \frac{e^{-\tau_2 - \tau_3}}{18} - \frac{e^{-\tau_1 - \tau_2 - \tau_3}}{18} + \frac{e^{-\tau_3}}{36} 
         = \bE\left[FP(x_5) \vert \cS'\right].
    \end{align*}
\end{itemize}

\paragraph{Species tree \texorpdfstring{$\cS''$}{S''} (Figure \ref{Fig_5LeavesAlternativeRankings}, bottom)}
\begin{itemize}
    \item FP indices on the species tree:
    \begin{align*}
        FP_{\cS''}(x_1) &= \tau_0 + \tau_1 + \tau_2 + \frac{\tau_3}{2} = FP_{\cS''}(x_2) \\
        FP_{\cS''}(x_3) &= \tau_0 + \tau_1 + \frac{\tau_2}{3} + \frac{\tau_3}{3} \\
        FP_{\cS''}(x_4) &= \tau_0 + \frac{\tau_1}{2} + \frac{\tau_2}{3} + \frac{\tau_3}{3} =  FP_{\cS''}(x_5).
    \end{align*}
    \item Expected FP indices:
   \begin{align*}
        \bE\left[FP(x_1) \vert \cS''\right] &= \tau_0 + \tau_1 + \tau_2 + \frac{\tau_3}{2} + 1 - \frac{e^{-\tau_2-2\tau_3}}{12} - \frac{e^{-\tau_1-\tau_2-2\tau_3}}{24} \\
        &\quad + \frac{e^{-\tau_2-\tau_3}}{36} + \frac{e^{-\tau_1-\tau_2-\tau_3}}{72} - \frac{e^{-\tau_3}}{12}
        =  \bE\left[FP(x_2) \vert \cS''\right] \\
        \bE\left[FP(x_3) \vert \cS''\right] &= \tau_0 + \tau_1 + \frac{\tau_2}{3} + \frac{\tau_3}{3} + 1 - \frac{e^{-\tau_1}}{9} - \frac{e^{-\tau_2-2 \tau_3}}{12} \\ 
        &\quad - \frac{e^{-\tau_2-\tau_3}}{9} + \frac{e^{-\tau_1-\tau_2-\tau_3}}{12} + \frac{e^{- \tau_3}}{18} \\
        \bE\left[FP(x_4) \vert \cS''\right] &= \tau_0 + \frac{\tau_1}{2} + \frac{\tau_2}{3} + \frac{\tau_3}{3} + 1 - \frac{e^{-\tau_1}}{9} - \frac{e^{-\tau_2-2 \tau_3}}{24} \\
        &\quad - \frac{e^{-\tau_1-\tau_2-2 \tau_3}}{24} + \frac{e^{-\tau_2-\tau_3}}{36} - \frac{e^{-\tau_1-\tau_2-\tau_3}}{18} + \frac{e^{-\tau_3}}{18} \\
        &= \bE\left[FP(x_5) \vert \cS''\right].
   \end{align*}
\end{itemize}

\section{Species name abbreviations used in Figure \ref{Fig_CombinedBoxplot}}
\begin{table}[htbp]
    \centering
    \caption{Species name abbreviations used in Figure \ref{Fig_CombinedBoxplot}.}
    \begin{tabular}{cll} \toprule
      Abbreviation & Scientific name & Common name \\ \midrule
      Bos & \textit{Bos taurus} & cattle \\
      Cal & \textit{Callithrix jacchus} & white-tufted-ear marmoset\\
      Can & \textit{Canis familiaris} & dog \\
      Cav & \textit{Cavia porcellus} & domestic guinea pig \\
      Cho & \textit{Choloepus hoffmanni} & Hoffmann's two-fingered sloth \\
      Das & \textit{Dasypus novemcinctus} & nine-banded armadillo \\
      Dip & \textit{Dipodomys ordii } & Ord's kangaroo rat \\
      Ech & \textit{Echinops telfairi} & lesser hedgehog tenrec \\
      Equ & \textit{Equus caballus} & domestic horse \\
      Eri & \textit{Erinaceus europaeus} & European hedgehog\\
      Fel & \textit{Felis catus} & cat \\
      Gal & \textit{Gallus gallus} & chicken \\
      Gor & \textit{Gorilla gorilla} & gorilla \\
      Hom & \textit{Homo sapiens} & human \\
      Lox & \textit{Loxodonta africana} & African bush elephant \\
      Mac & \textit{Macropus eugenii} & tammar wallaby\\
      Mic & \textit{Microcebus murinus} & gray mouse lemur \\
      Mon & \textit{Monodelphis domestica} & gray short-tailed opossum \\
      Mus & \textit{Mus musculus} & house mouse \\
      Myo & \textit{Myotis lucifugus} & little brown bat \\
      New & \textit{Macaca mulatta} & rhesus macaque\\
      Och & \textit{Ochotona princeps} & American pika \\
      Orn & \textit{Ornithorhynchus anatinus} & platypus \\
      Ory & \textit{Oryctolagus cuniculus} & European rabbit \\
      Oto & \textit{Otolemur garnettii} & northern greater galago\\
      Pan & \textit{Pan troglodytes} & cimpanzee \\
      Pon & \textit{Pongo pygmaeus} & Bornean orangutan \\
      Pro & \textit{Procavia capensis} & rock hyrax\\
      Pte & \textit{Pteropus vampyrus} & large flying fox \\
      Rat & \textit{Rattus} & rat \\
      Sor & \textit{Sorex araneus} & common shrew \\
      Spe & \textit{Spermophilus tridecemlineatus} &  thirteen-lined ground squirrel \\
      Sus & \textit{Sus scrofa} & wild boar \\
      Tar & \textit{Tarsius syrichta} & Philippine tarsier\\
      Tup & \textit{Tupaia belangeri} & northern treeshrew \\
      Tur & \textit{Tursiops truncatus} & bottlenose dolphin \\
      Vic & \textit{Vicugna pacos} & alpaca \\
      \bottomrule
    \end{tabular}
    \label{tab_names}
\end{table}
\clearpage

\section{Details on species tree estimation for Section~\ref{Sec_Empirical}}
In order to obtain species tree estimates for the multilocus mammal data set of \cite{Song2012}, we downloaded the full sequences of the 447 genes from \url{https://datadryad.org/stash/dataset/doi:10.5061\%2Fdryad.3629v} (file \texttt{gene447\_final.nexus)}.

After converting to PHYLIP format, a maximum liklihood tree based on the concatenated data was obtained using RAxML version 8.2.10 \citep{stamatakis-raxml} via the following command:\\
\texttt{raxMLHPC-AVX -s gene447\_final.phy -p 12345 -m GTRGAMMA -o Gal -n raxml\_GTR.tre}\\

Afterwards, maximum likelihood branch lengths under the GTR+$\Gamma$ model and enforcing a molecular clock were computed in PAUP* \citep{swofford-paup} via the following commands:\\

\noindent \texttt{exe gene447\_final.nexus} \\
\texttt{gettrees file=raxml\_GTR.tre} \\
\texttt{set criterion=likelihood} \\
\texttt{lset nst=6 basefreq=empirical rates=gamma rmatrix=estimate shape=estimate clock=yes} \\
\texttt{lscores} \\
\texttt{savetrees RAxML\_GTR\_MLbranchlengths.tre brlens}\\

Moreover, a species tree using the species tree estimation method SVDQuartets \citep{Chifman2014} and maximum likelihood branch lengths for this species tree were obtained in PAUP* via the following commands:\\

\noindent\texttt{exe gene447\_final.nex}\\
\texttt{svdq}\\
\texttt{outgroup 35} (where species 35 corresponds to Gal)\\
\texttt{roottrees rootmethod=outgroup}\\
\texttt{set criterion=likelihood} \\
\texttt{lset nst=6 basefreq=empirical rates=gamma rmatrix=estimate shape=estimate clock=yes} \\
\texttt{lscores} \\
\texttt{savetrees svdq\_GTR\_MLbranchlengths.tre brlens}\\

\end{document}